%% file: main.tex
\documentclass[a4paper,english,cleveref]{lipics-v2021}
\usepackage{xspace}
\usepackage{complexity}
\usepackage{siunitx}
\usepackage{thmtools}

\usepackage[subrefformat=simple,labelformat=simple]{subcaption}

\crefname{claim}{Claim}{Claims}

\hideLIPIcs
\nolinenumbers

\title{Connected Coordinated Motion Planning with Bounded Stretch}
\titlerunning{Connected Coordinated Motion Planning with Bounded Stretch}
\author{Sándor P. Fekete}{Department of Computer Science, TU Braunschweig, Braunschweig,  Germany}{s.fekete@tu-bs.de}{https://orcid.org/0000-0002-9062-4241}{}
\author{Phillip Keldenich}{Department of Computer Science, TU Braunschweig, Braunschweig,  Germany}{p.keldenich@tu-bs.de}{https://orcid.org/0000-0002-6677-5090}{}
\author{Ramin Kosfeld}{Department of Computer Science, TU Braunschweig, Braunschweig,  Germany}{r.kosfeld@tu-bs.de}{https://orcid.org/0000-0002-1081-2454}{}
\author{Christian Rieck}{Department of Computer Science, TU Braunschweig, Braunschweig, Germany}{rieck@ibr.cs.tu-bs.de}{https://orcid.org/0000-0003-0846-5163}{}
\author{Christian Scheffer}{Faculty of Electrical Engineering and Computer Science, Bochum University of Applied Sciences, Bochum,  Germany}{christian.scheffer@hs-bochum.de}{https://orcid.org/0000-0002-3471-2706}{}

\authorrunning{S.~P.~Fekete, P.~Keldenich, R.~Kosfeld, C.~Rieck, and C.~Scheffer}
\Copyright{Sándor P. Fekete, Phillip Keldenich, Ramin Kosfeld, Christian Rieck, Christian Scheffer}

\ccsdesc{Theory of computation~Computational geometry}
\ccsdesc{Computing methodologies~Motion path planning}

\keywords{Motion planning, parallel motion, bounded stretch, scaled shape, makespan, connectivity, swarm robotics}

\relatedversion{This is the full version of an extended abstract~\cite{unlabeled-connected-isaac} that appeared in the 32nd International Symposium on Algorithms and Computation (ISAAC 2021), December 6-8, 2021.} 

\supplement{}

\funding{} 

\acknowledgements{}

\EventEditors{}
\EventNoEds{0}
\EventLongTitle{}
\EventShortTitle{}
\EventAcronym{}
\EventYear{}
\EventDate{}
\EventLocation{}
\EventLogo{}
\SeriesVolume{}
\ArticleNo{}

\newcommand{\mountain}{\mathcal{M}\xspace}
\newcommand{\valley}{\mathcal{V}\xspace}

\begin{document}

\maketitle

\begin{abstract}
We consider the problem of connected coordinated motion planning for a large collective of simple, identical robots: 
From a given start grid configuration of robots, we need to reach a desired target configuration via a sequence of parallel, collision-free robot motions, such that the set of robots induces a connected grid graph at all integer times. 
The objective is to minimize the \emph{makespan} of the motion schedule, i.e., to reach the new configuration in a minimum amount of time. 
We show that this problem is \NP-complete, even for deciding whether a makespan of $2$ can be achieved, while it is possible to check in polynomial time whether a makespan of $1$ can be achieved.
On~the algorithmic side, we establish simultaneous constant-factor approximation for two fundamental parameters, by achieving \emph{constant stretch} for \emph{constant scale}.
Scaled shapes (which arise by increasing all dimensions of a given object by the same multiplicative factor) have been considered in previous seminal work on self-assembly, often with unbounded or logarithmic scale factors; we provide methods for a generalized scale factor, bounded by a constant.
Moreover, our algorithm achieves a \emph{constant stretch factor}:
If mapping the start configuration to the target configuration requires a maximum Manhattan distance of $d$, then the total duration of our overall schedule is $\mathcal{O}(d)$, which is optimal up to constant~factors.
\end{abstract}

\input{01-introduction-full}
\input{02-preliminaries-full}
\input{03-hardness-full}
\input{04-algorithm-full}
\input{05-conclusion-full}

\bibliography{references}

\end{document}

%% file: 01-introduction-full.tex
\section{Introduction}
Coordinating the motion of a set of objects is a fundamental problem that occurs in a large spectrum of theoretical contexts and practical applications. 
A~typical challenge arises from relocating a large collection of agents from a given start into a desired goal configuration in an efficient manner, while respecting a number of natural constraints, such as avoiding collisions or disrupting the coherence of the arrangement. 
This problem was also the subject of the 2021~Computational Geometry Challenge, highlighting the high relevance for the algorithmic community; see~\cite{FeketeKKM22} for an overview and~\cite{shadoks,gitastrophe,UNIST} for successful contributions.

In this paper, we consider a \emph{connected} configuration of objects, e.g., a (potentially large) collective of mobile robots or blocks of building material that can be moved by a large group of robots, which needs to be transformed into a desired target configuration by a sequence of parallel, collision-free motions that keeps the overall arrangement connected at all integer times.
Problems of this type occur in many contexts requiring relocation of autonomous agents; the connectivity constraint arises naturally, e.g., for assemblies in space, where disconnected pieces cannot regain connectivity, or for small-scale swarm robots (such as \emph{catoms} in \emph{claytronics}~\cite{goldstein2004claytronics}) which need connectivity for local motion, electric power and communication; see~\cref{fig:catoms}.

\begin{figure}[ht]
	\centering
	\includegraphics[height=0.275\textwidth]{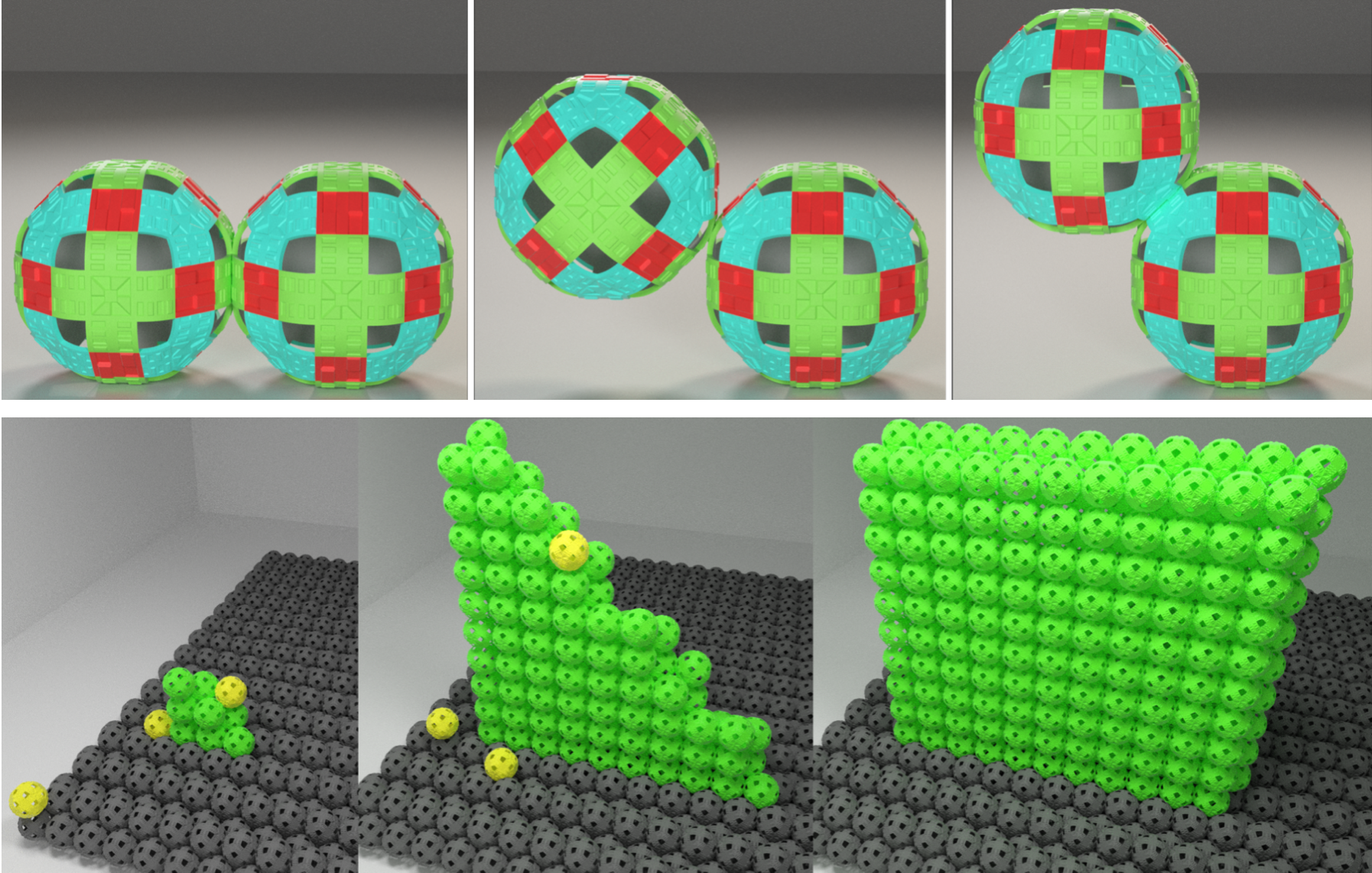}
	\includegraphics[height=0.275\textwidth]{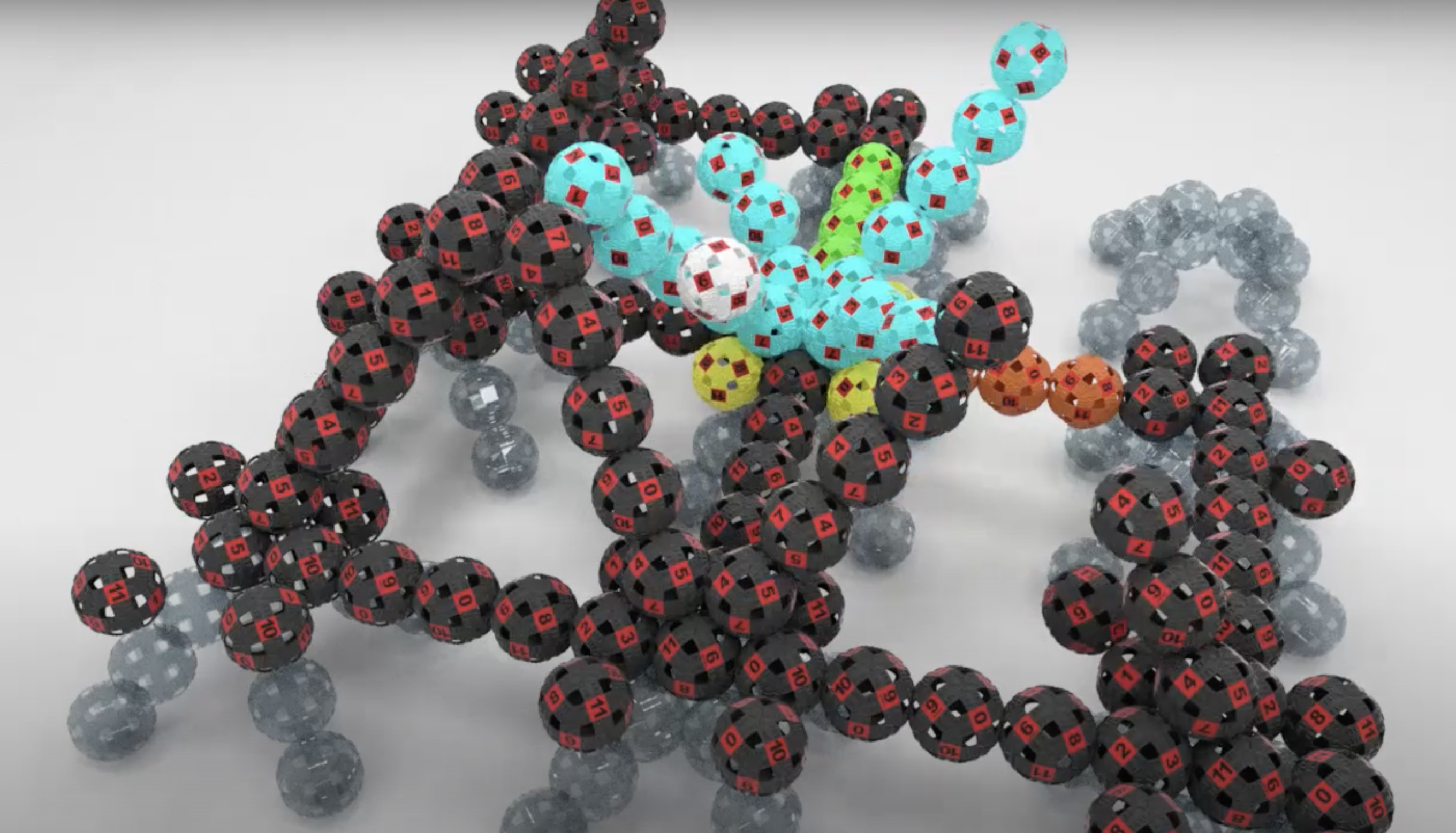}
	\caption{
		(Top left) An autonomous, sphere-shaped catom, changing location by rotating around a second catom used as a pivot~\cite{bourgeois2020}.
		(Bottom left) A collective of catoms building a wall~\cite{bourgeois2020}.
		(Right)~A~configuration of catoms in the process of building a scaffold structure~\cite{thalamy2019distributed}.}
	\label{fig:catoms}
\end{figure}

A crucial algorithmic aspect is \emph{efficiency}: How can we coordinate the robot motions, such that a target configuration is reached in timely or energy-efficient manner? 
Most previous work has largely focused on sequential schedules, where one robot moves at a time, with objectives such as minimizing the number of moves. 
In practice, however, robots usually move simultaneously, so we desire a \emph{parallel} motion schedule, with a natural objective of minimizing the time until completion, called \emph{makespan}.
How well can we exploit parallelism in a robot swarm to achieve an efficient schedule?
As illustrated in~\cref{fig:example}, this is where the connectivity constraints make a tremendous difference.

A critical parameter in self-assembly is the robustness of the involved shapes, corresponding to sufficient local connectivity to prevent fragility. 
This leads to the concept of \emph{scaled shapes}; intuitively, a scale factor of $c$ corresponds to replacing each pixel of a polyomino shape by a quadratic $c\times c$ array of pixels. 
This has fundamental connections to Kolmogorov and runtime complexity, as shown by Soloveichik and Winfree~\cite{soloveichik2007complexity}:
``Furthermore, the independence of scale in self-assembly theory appears to play the same crucial role as the independence of running time in the theory of computability$\ldots$ [we] show that the running-time complexity, with respect to Turing machines, is polynomially equivalent to the scale complexity of the same function implemented via self-assembly by a finite set of tile types.''
As a consequence, limiting scale has received considerable attention: While Soloveichik and Winfree established unbounded scale in general self-assembly, other work has managed to achieve logarithmic and even (in specific scenarios) constant scale.

As we demonstrate in this paper, achieving optimal makespan for connected reconfiguration is provably hard, even in relatively basic cases. 
On the positive side, we present methods that are capable of achieving a constant-factor approximation, assuming not more than a generalization of constant scale of start and target configurations. 
In fact, our method realizes \emph{constant stretch}: If mapping the start configuration to the target configuration requires a maximum Manhattan distance of $d$, then the total duration of our overall schedule is~$\mathcal{O}(d)$. 
As can be seen from \cref{fig:example-a} (where $d$ corresponds to trajectory A), this is less straightforward than in a non-connected setting, even in very basic instances. 
Instead, this quickly requires coordination of the whole arrangement, as sketched in~\cref{fig:example-b}.

\begin{figure}[ht]
	\centering
	\begin{subfigure}[b]{0.24\textwidth}
		\centering
		\includegraphics[width=0.8\textwidth]{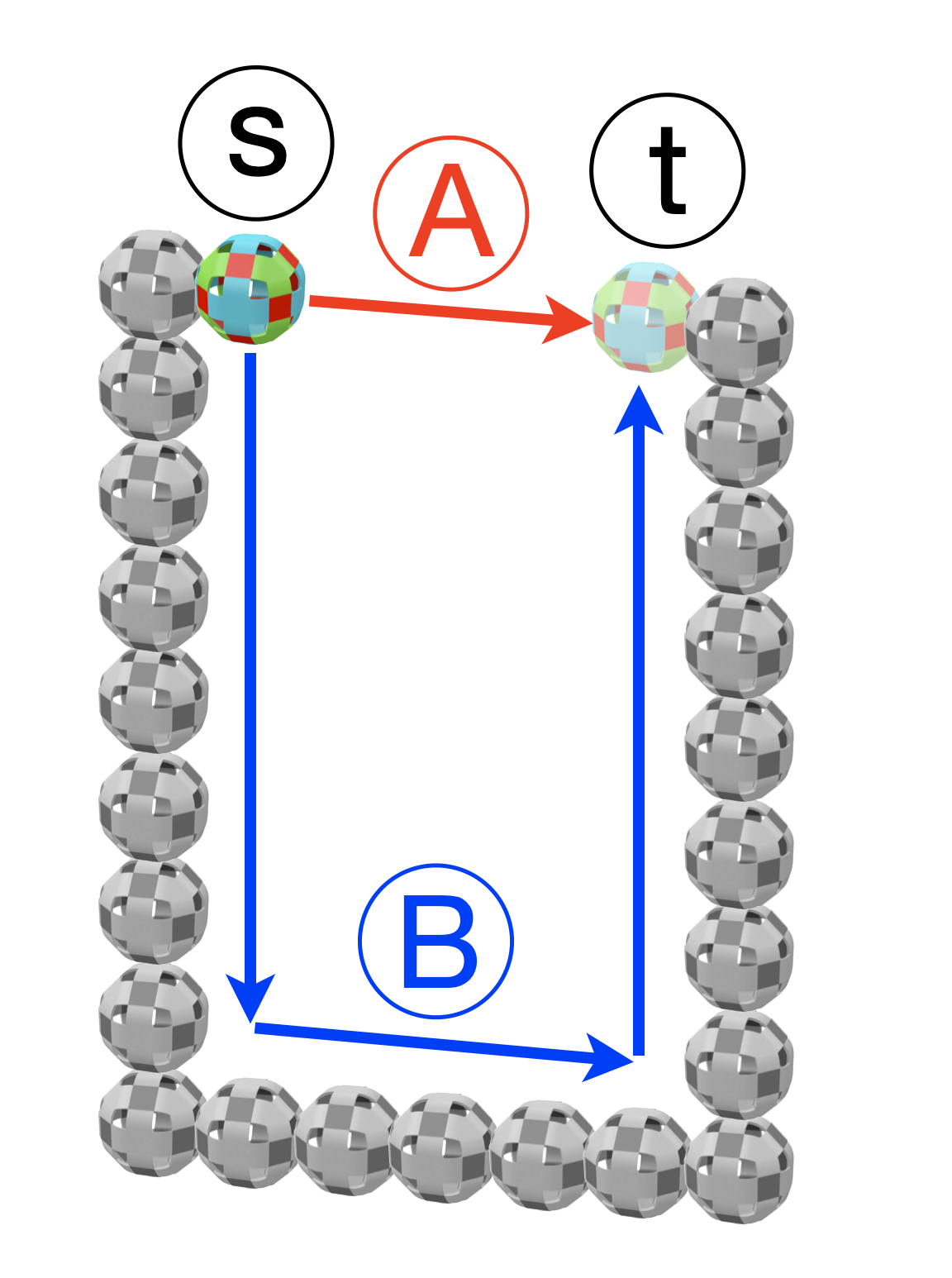}
		\caption{}
		\label{fig:example-a}
	\end{subfigure}\hfil
	\begin{subfigure}[b]{0.24\textwidth}
		\centering
		\includegraphics[width=0.8\textwidth]{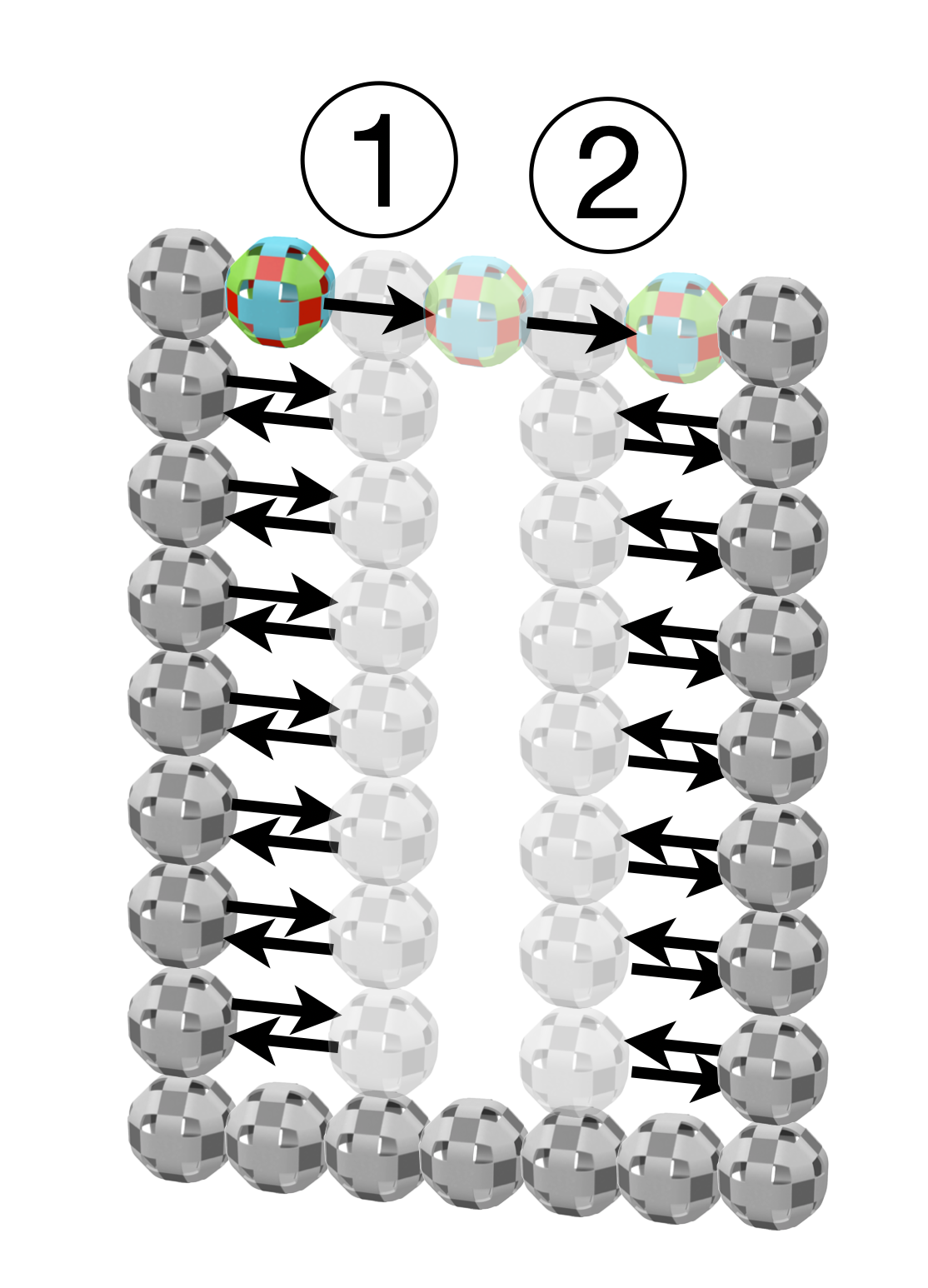}
		\caption{}
		\label{fig:example-b}
	\end{subfigure}\hfil
	\begin{subfigure}[b]{0.24\textwidth}
		\centering
		\includegraphics[width=0.8\textwidth]{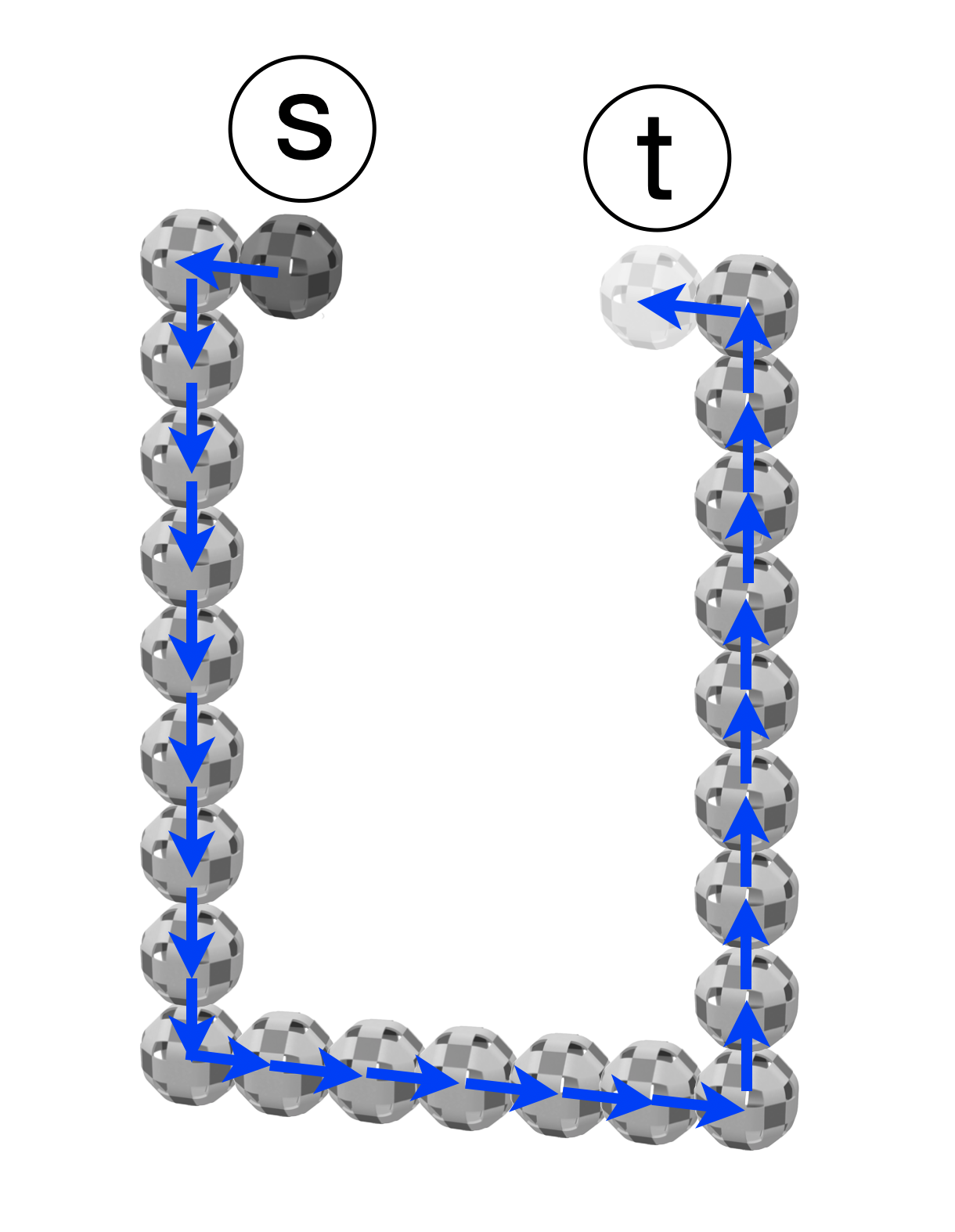}
		\caption{}
		\label{fig:example-c}
	\end{subfigure}
	\begin{subfigure}[b]{0.24\textwidth}
		\centering
		\includegraphics[width=\textwidth]{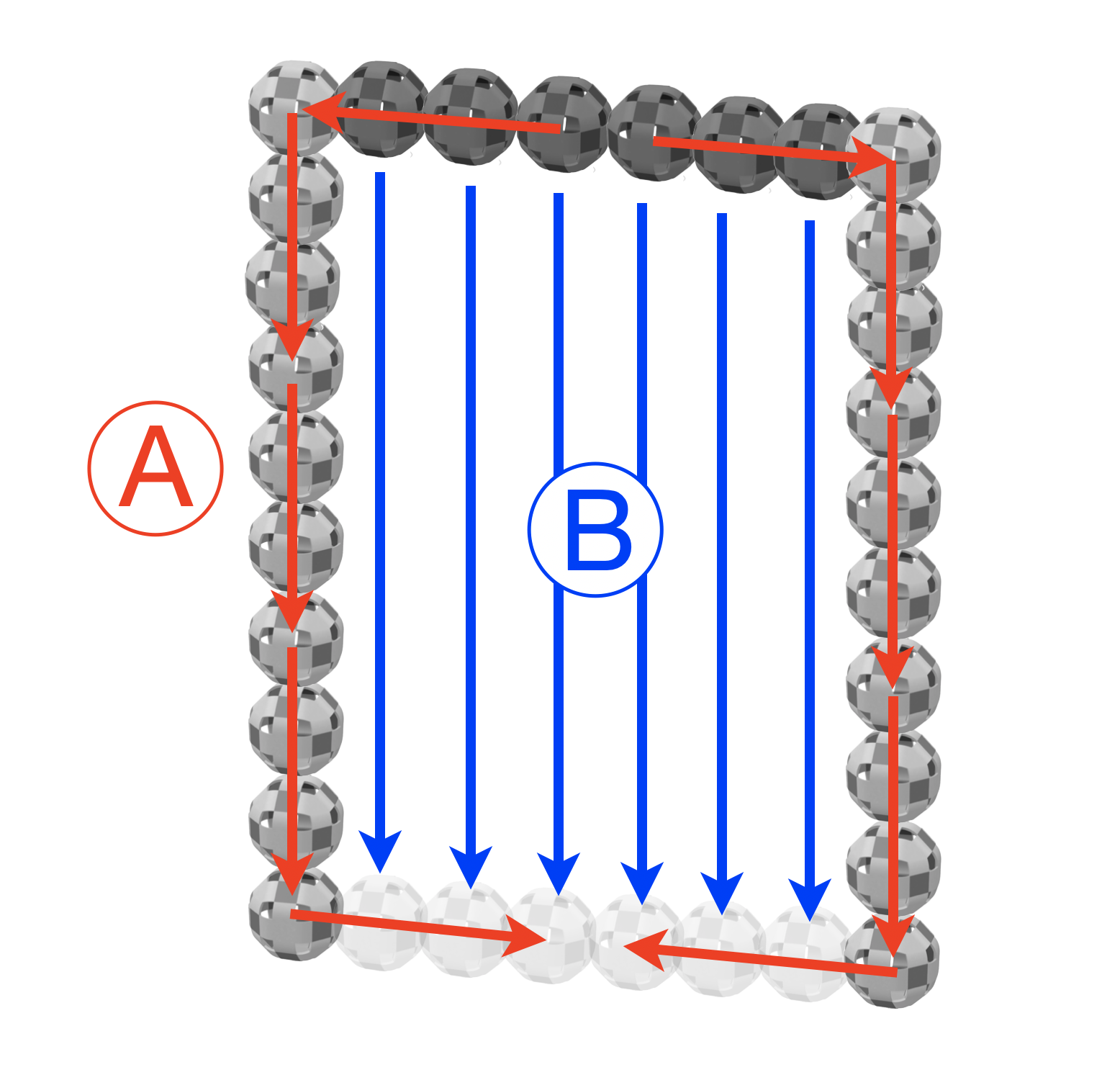}
		\caption{}
		\label{fig:example-d}
	\end{subfigure}
	\caption{
		Reconfiguration with and without connectivity constraints.
		(a) Relocating the colored particle from $s$ to $t$, without (red trajectory A) and with connectivity constraint (blue trajectory B).
		(b) Coordinating many particles to quickly deliver a specific particle to a desired location, while preserving connectivity. 
		Both vertical lines of robots move in parallel towards each other, staying connected to the horizontal line.
		(c) Reconfiguring an arrangement of \emph{identical} particles in a single, parallel, connected step. 
		In this case, it is possible that the whole arrangement rotates one step, as the robots are indistinguishable; the top-rightmost robot moves onto the goal location and all other robots follow.
		(d) Reconfiguring an arch-shaped arrangement of identical particles into a U-shaped one, without (motion plan A, shown in red) and with connectivity (motion plan B, shown in blue). 
		Motion plan A splits the configuration, resulting in a makespan of length half the amount of robots in the top horizontal row. 
		In motion plan B everything stays connected, but the horizontal row has to move all the way to the bottom, what results in longer makespan.}
	\label{fig:example}
\end{figure}

\subsection{Our Results}
We provide a spectrum of new results for questions arising from efficiently reconfiguring a connected, unlabeled collective of robots from a start configuration~$C_s$ into a target configuration $C_t$, aiming for minimizing the overall makespan and maintaining connectivity in each step.
We obtained all of our results in obstacle-free \num{2}D space, and discuss possible extensions as well as open questions in the conclusion.
Some of these results are featured in~a~video~\cite{unlabeled-connected-video}.

\begin{itemize}
	\item We show that deciding whether there is a schedule with a makespan of $1$ transforming $C_s$ into $C_t$ can be done in polynomial time, see~\cref{lem:single-move-reconfiguration}.
	\item We show that deciding whether there is a schedule with a makespan of $2$ transforming $C_s$ into $C_t$ is \NP-complete, see~\cref{thm:connected-motion-planning-hard}. 
	This implies \NP-hardness of approximating the minimum makespan within a constant of $(\frac{3}{2}-\varepsilon)$, for any $\varepsilon>0$, see~\cref{cor:connected-motion-planning-optmial-hard}.
	\item As our main result, we show that there is a constant $c^*$ such that for any pair of start and target configurations with a (generalized) scale of at least $c^*$, a schedule with constant stretch can be computed in polynomial time, see~\cref{thm:scaled-instances-bounded-stretch,,cor:non-overlapping-instances}. 
	This~implies that there is a constant-factor approximation for the problem of computing schedules with minimal makespan restricted to pairs of start and target configurations with a scale of at least $c^*$, see~\cref{cor:scaled-instances-constant-factor-approx}.
\end{itemize}

\subsection{Related Work}
Coordinating the motion of many agents plays a central role when dealing with large numbers of moving robots, vehicles, aircraft, or people.
How can each agent choose an efficient route that avoids collisions with other agents as they simultaneously move to their destinations?
These basic questions arise in many applications, such as ground swarm robotics~\cite{rubenstein2014programmable,sw-sr-08}, aerial swarm robotics~\cite{cpdsk-saesr-18,kumar}, air traffic control~\cite{delahaye2014mathematical}, and vehicular traffic networks~\cite{fhtwhfe-mift-11,ss-hbtn-04}.

Multi-robot coordination dates back to the early days of robotics and computational geometry.
The seminal work by Schwartz and Sharir~\cite{ss-pmpcbpb-83} from the 1980s considers coordinating the motion of disk-shaped objects among obstacles. 
Their algorithms are polynomial in the complexity of the obstacles, but exponential in the number of disks.
Hopcroft et al.~\cite{hss-cmpmio-84} and Hopcroft and Wilfong~\cite{hw-rmompgs-86} proved \PSPACE-completeness of moving multiple robots to a target configuration, showing the significant challenge of coordinating many robots.

There is a vast body of other related work dealing with multi-robot motion planning, both from theory and practice. 
For a more extensive overview, see~\cite{dfk+-arxiv}. 
In both discrete and geometric variants of the problem, the objects can be \emph{labeled}, \emph{colored} or \emph{unlabeled}.
In the \emph{labeled} case, the objects are all distinguishable and each object has its own, uniquely defined target position.
In the \emph{colored} case, the objects are partitioned into $k$ groups and each target position can only be covered by an object with the right color.
This was considered by Solovey and Halperin~\cite{sh-kcmrmp-14}, who present and evaluate a practical sampling-based algorithm.
In the \emph{unlabeled} case, objects are indistinguishable and target positions can be covered by any object.
This was first considered by Kloder and Hutchinson~\cite{kh-ppimf-06}, who presented a practical sampling-based algorithm.
Turpin~et~al.~\cite{tmk-tpams-13} give an algorithm for finding a solution in polynomial time, if one exists.
This is optimal with respect to the longest distance traveled by any one robot, but only holds for disk-shaped robots under additional restrictive assumptions on the free space.
For unit disks and simple polygons, Adler et al.~\cite{adh+-emmpudsp-15} provide a polynomial-time algorithm under the additional assumption that the start and target positions have some minimal distance from each other.
Under similar separability assumptions, Solovey et al.~\cite{syz+-mpudog-15} provide a polynomial-time algorithm that produces a set of paths that is no longer than $\mbox{OPT}+4m$, where $m$ is the number of robots, and $\mbox{OPT}$ denotes the total length of a set of paths of an optimal solution.
However, they do not consider the makespan, but only the total path length.
On the negative side, Solovey and Halperin~\cite{sh-hummp-15} prove that the unlabeled multiple-object motion planning problem is \PSPACE-hard, even when restricted to unit square objects in a polygonal environment.

For an extensive overview of multi-agent path planning, refer to~\cite{SternSFK0WLA0KB19}.
Yu and LaValle~\cite{YuL12} discuss the relationship of multi-agent path planning and flow problems in \emph{collision-free unit-distance graphs}. 
These are unit-distance graphs having the additional property that two discs of radius $\sqrt{2}/4$ do not collide when traveling with unit speed on two paths that do not contain the same vertex for the same time step.
They consider different problems settings and show, among other results, that if the goal locations can be assigned to arbitrary agents, a solution always exists and the longest path has at most $n+V-1$ edges, where $n$ and $V$ are the number of agents and vertices of the graph, respectively.
They also mention \NP-hardness of the decision problem in case goals are pre-assigned to agents.
Charrier et al.~\cite{CharrierQSS19,CharrierQSS19A,CharrierQSS20} study reachability and coverage planning problems for connected agents. 
They show different complexity results for different \emph{topological graphs}, consisting of vertices and two kinds of edges, one for robot motion, the other for communication. 
They also introduce \emph{sight-movable graphs}, which are undirected topological graphs for which there is a movable path between any pair of vertices if these vertices can communicate with each other.
They show that this class admits efficient algorithms for the reachability and coverage problem.
Queffelec, Sankur, and Schwarzentruber~\cite{QueffelecSS23} study the connected multi-agent path finding problem in partially known environments in which the  graph is not known entirely in advance, and show \PSPACE-completeness of the problem.
Despite the fact that all of this is related to our work, a crucial difference is that we consider the stretch factor as the main performance measure.

There is also a wide range of practical related work.
Self-configuration of robots as active agents was studied by Naz et al.~\cite{naz2016distributed}.
A basic model in which robots are used as building material was introduced by Derakhshandeh et al.~\cite{derakhshandeh2015algorithmic,derakhshandeh2016universal}.
This resembles Claytronics robots like Catoms, see Goldstein and Mowry~\cite{goldstein2004claytronics}.
In more recent work, Thalamy et al.~\cite{thalamy2019distributed} consider using scaffolding structures for asynchronous reconfiguration.

For an instance of parallel reconfiguration, a lower bound for the time required for \emph{all} robots to reach their destinations is the time it takes to move just \emph{one} robot to its destination in the absence of other robots, i.e., by the maximum distance between a robot's origin and destination.
Moving a dense arrangement of robots to their destinations while avoiding collisions may require substantially more time than this lower bound.
This motivates the \emph{stretch factor}, which is defined to be the ratio of the time taken by a parallel motion plan divided by the simple lower bound.

In recent work, Demaine et al.~\cite{dfk+-arxiv} provide several fundamental insights into these problems of coordinated motion planning for the scenario with labeled robots without a connectivity constraint. 
They were able to develop algorithms that (under relatively mild assumptions on the separation between robots) can achieve \emph{constant} stretch factors that are independent of the number of robots. 
Thus, these algorithms provide an absolute performance guarantee on the makespan of the parallel motion schedule, which implies that the schedule is a constant-factor approximation of the best possible schedule.
For densely packed arrangements of robots (without separation assumptions), they prove that a constant stretch factor is no longer possible, and give upper and lower bounds on the worst-case stretch factor. 
Note that the approaches of~\cite{dfk+-arxiv} cannot be adapted, thus our demand for connectivity requires new algorithmic ideas.

\smallskip
In the methods developed in~\cite{coordinated_video,dfk+-arxiv}, elementary pieces can achieve arbitrary relative configurations, as shown in \cref{fig:coordinated}, in which colors indicate the final destinations of~components. 

\begin{figure}[ht]
	\centering
	\begin{subfigure}[b]{0.24\textwidth}
		\centering
		\includegraphics[width=0.95\textwidth]{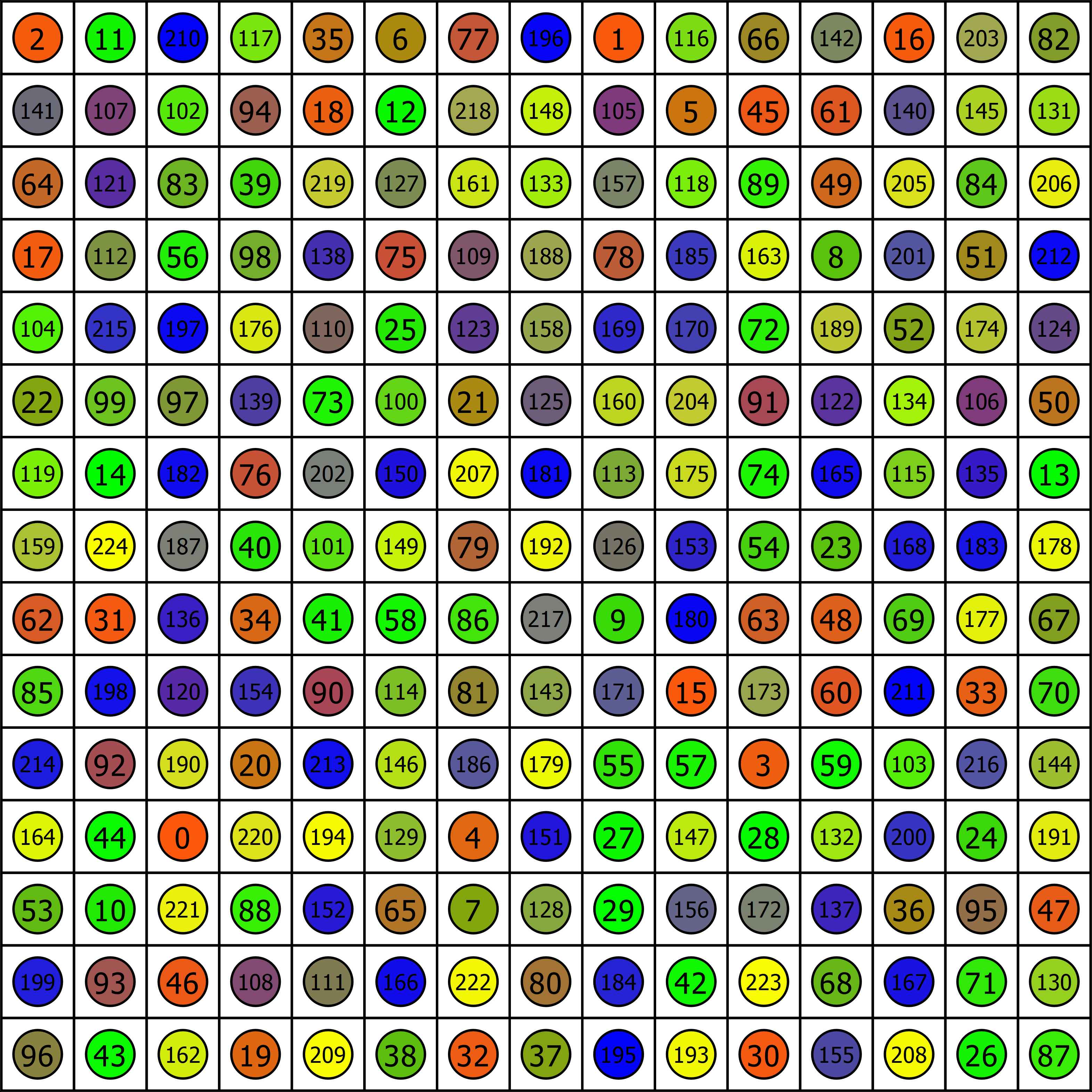}
		\caption{}
	\end{subfigure}\hfil
	\begin{subfigure}[b]{0.24\textwidth}
		\centering
		\includegraphics[width=0.95\textwidth]{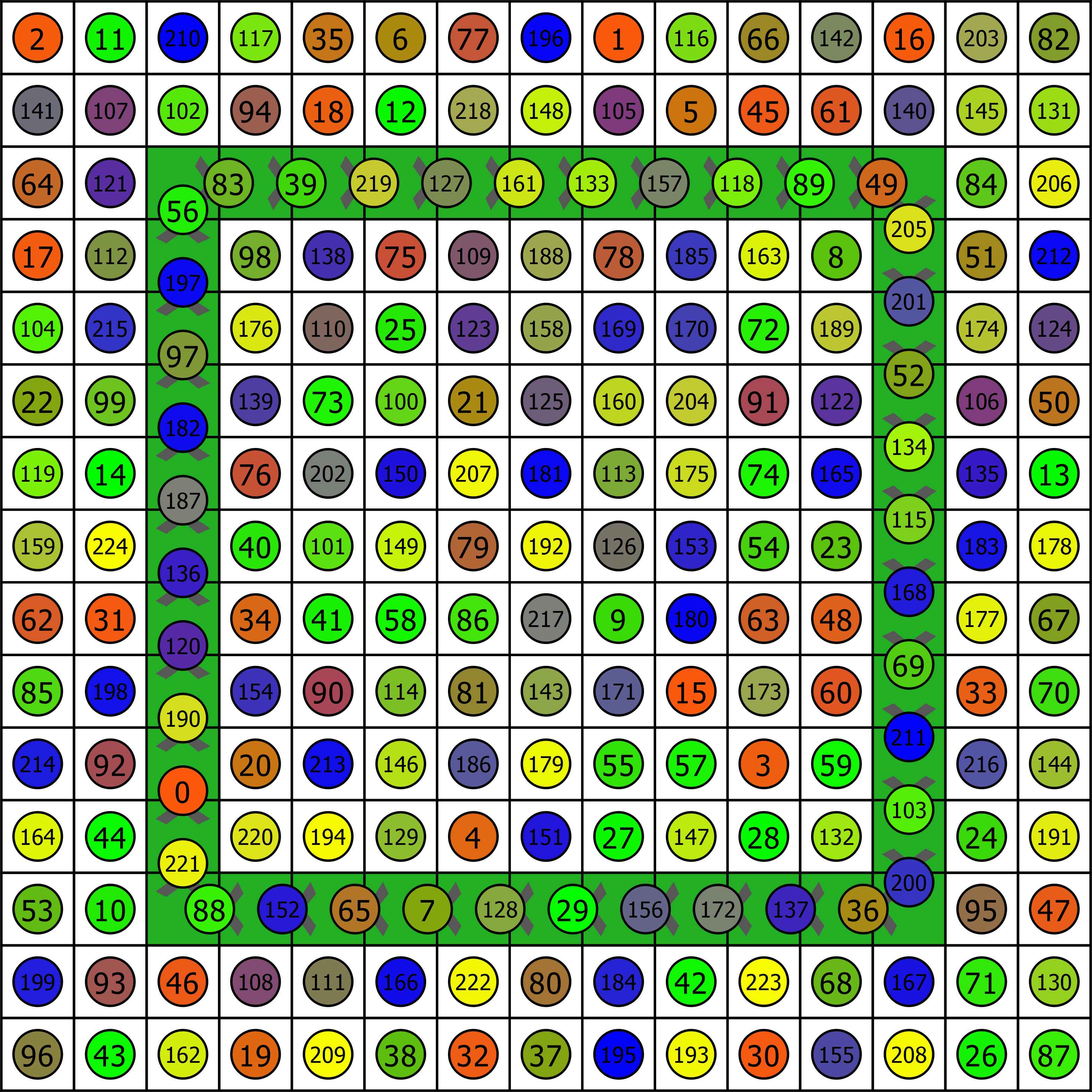}
		\caption{}
	\end{subfigure}\hfil
	\begin{subfigure}[b]{0.24\textwidth}
		\centering
		\includegraphics[width=0.95\textwidth]{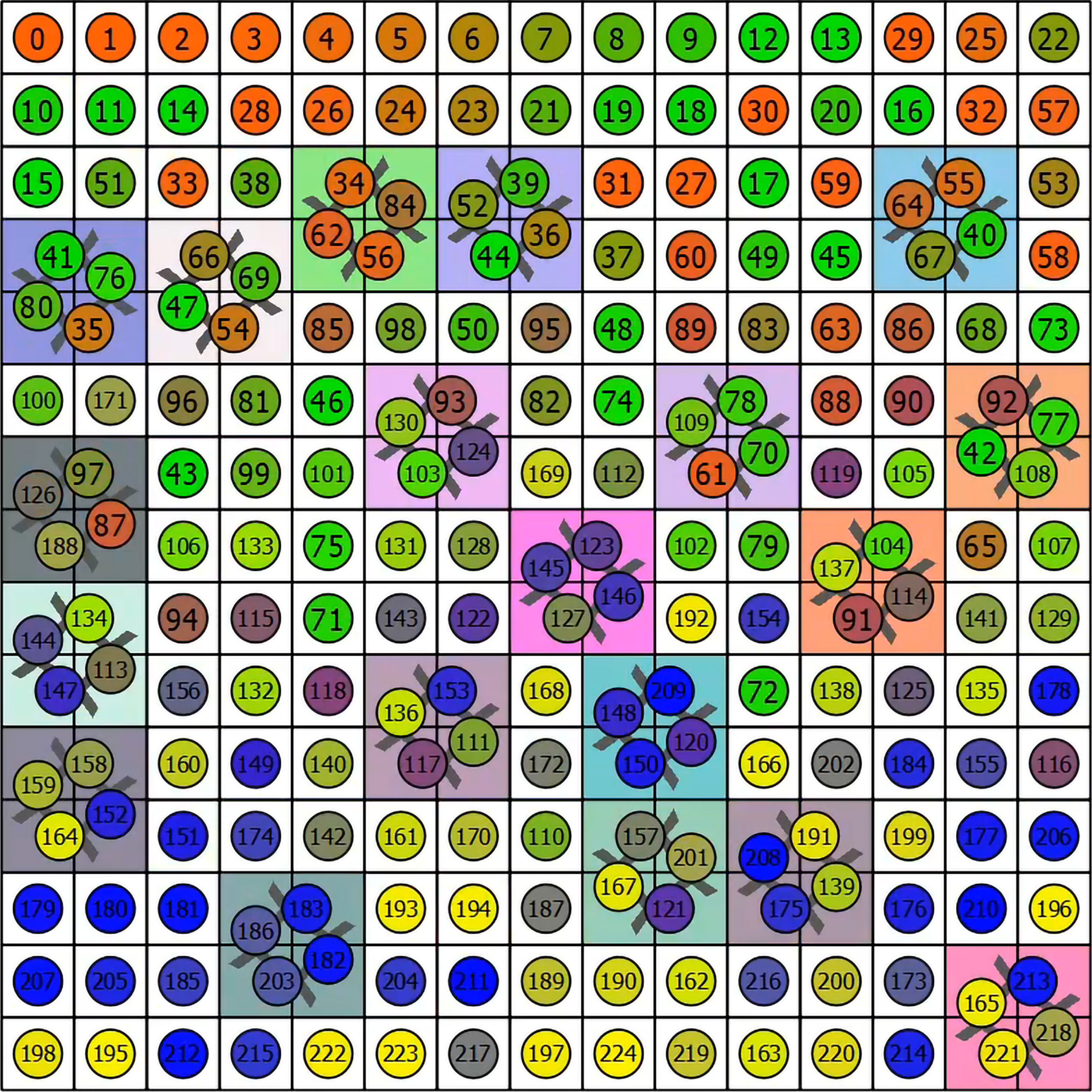}
		\caption{}
	\end{subfigure}\hfil
	\begin{subfigure}[b]{0.24\textwidth}
		\centering
		\includegraphics[width=0.95\textwidth]{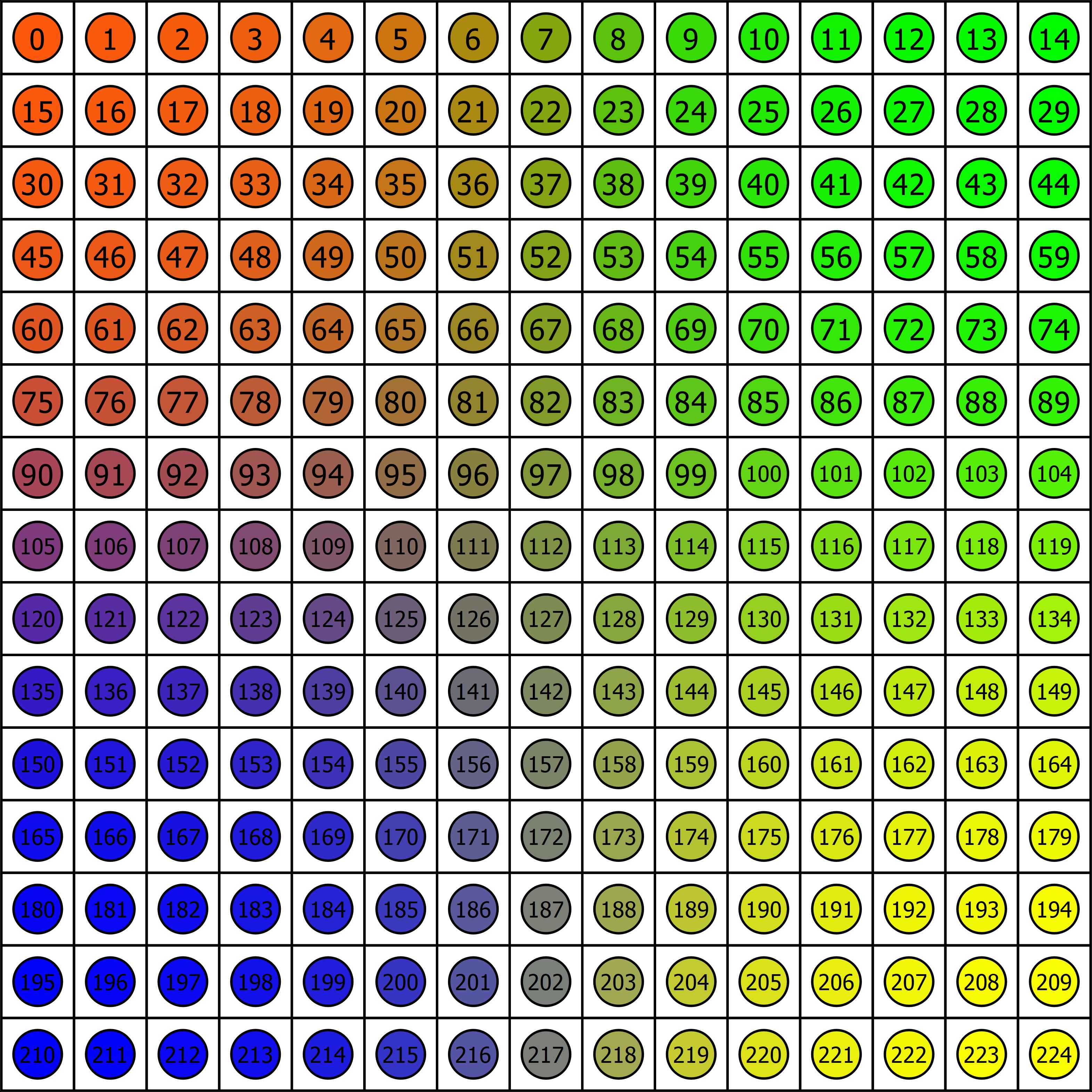}
		\caption{}
	\end{subfigure}
	\caption{Parallel reconfiguration, established by~\cite{dfk+-arxiv}:
		(a) Start configuration. 
		(b) A feasible, parallel reconfiguration move. 
		(c) Parallel reconfiguration moves. 
		(d) Target configuration.
		See \protect\url{https://www.ibr.cs.tu-bs.de/users/fekete/Videos/CoordinatedMotionPlanning.mp4} for a video~\cite{coordinated_video}.}
	\label{fig:coordinated}
\end{figure}

As discussed in~\cite{coordinated_video,dfk+-arxiv}, it is straightforward to adapt this to geometric reconfigurations: ``Filled'' pixels (corresponding to material) get assigned one color (e.g., ``black''), while ``empty'' pixels (corresponding to void) get a second color (e.g., ``white''). 
However, in real-world scenarios (such as in space, or for swarm robots as shown in \cref{fig:catoms,fig:example}), basic components need to stay connected to keep them from drifting apart. 
How can we develop efficient parallel \emph{connected} reconfiguration schedules, in which all ``black'' pixels remain connected throughout the process? 
(Note that we aim for efficient schedules, i.e., \emph{constant stretch}, so approaches like the recent one by Akitaya et al.~\cite{musk21} for transforming any arrangement into the same base configuration are insufficient.)

The concept of scale complexity has received a considerable amount of attention in the context of self-assembly; in all settings, achieving constant scale has required special cases or operations. 
As shown by Soloveichik and Winfree~\cite{soloveichik2007complexity}, the minimal number of distinct tile types necessary to self-assemble a shape, at some scale, can be bounded both above and below in terms of the shape’s Kolmogorov complexity, leading to unbounded scale in general. 
As shown by Demaine et al.~\cite{demaine2011self}, allowing additional operations that allow destroying tiles can be exploited to achieve a scale that is only bounded by a logarithmic factor, beating the linear bound without such operations.
In a setting of recursive, multi-level \emph{staged} assembly with a logarithmic number of stages (i.e., ``hands'' for handling subassemblies), Demaine~et~al.~\cite{ddf-ssnas-08} achieved logarithmic scale, and constant scale for more constrained classes of polyomino shapes; this was later improved by Demaine et al.~\cite{dfs+-ngafc-17} to constant scale for a logarithmic number of stages.
More recently, Luchsinger et al.~\cite{luchsinger2019self} employed repulsive forces between tiles to achieve constant scale in two-handed self-assembly.

%% file: 02-preliminaries-full.tex
\section{Preliminaries}
We consider \emph{robots} at integer grid positions.
A set of $n$ unlabeled robots forms a \emph{configuration}~$C$, corresponding to a vertex-induced subgraph $H$ of the infinite integer grid, with an edge between two grid vertices $v_1,v_2 \in C$ if and only if $v_1$ and $v_2$ are on adjacent grid positions, i.e., a distance of 1 apart.  
A configuration is \emph{connected}, if $H$ is connected.  
Two~configurations~$C_1$ and $C_2$ \emph{overlap}, if they have at least one position in common. 
Two~robots are \emph{adjacent} if their positions $v_1,v_2$ are adjacent, and \emph{diagonally adjacent} if their positions are adjacent with a common vertex $v$ such that $(v_1,v)$ and $(v,v_2)$ lie orthogonal.

A robot can {move} in discrete time steps by changing its location from a grid position $v$ to an adjacent grid position $w$; denoted by $v \rightarrow w$. Two moves $v_1 \rightarrow w_1$ and $v_2 \rightarrow w_2$ are called \emph{collision-free} if $v_1 \neq v_2$ and $w_1 \neq w_2$. 
A \emph{transformation} between two configurations $C_1 = \{ v_1,\dots,v_n\}$ and $C_2 = \{ w_1,\dots,w_n \}$ is a set of collision-free moves $\{ v_i \rightarrow w_i \mid i = 1,\dots,n \}$. 
Furthermore, robots are allowed to remain in their current position, i.e., we allow the move $v\rightarrow v$.
However, even though this is unproblematic for the unlabeled scenario, we do not allow \emph{swaps} between robots, i.e., the two parallel moves $v_1\rightarrow v_2$ and $v_2\rightarrow v_1$ cause a~collision.
	
For $M \in \mathbb{N}$, a \emph{schedule} is a sequence $C_1 \rightarrow \cdots \rightarrow C_{M+1}$ (also denoted as $C_1 \rightrightarrows C_{M+1}$) of transformations, with a \emph{makespan} of~$M$. 
A \emph{stable schedule} $C_1 \rightrightarrows_{\chi} C_{M+1}$ uses only connected configurations.  
Let $C_{s}, C_{t}$ be two connected configurations with equally many robots called \emph{start} and \emph{target configuration}, respectively. 
A \emph{matching} is a one-to-one mapping between vertices from $C_s$ and~$C_t$. 
The \emph{diameter} of a matching is the maximal Manhattan distance between two matched vertices.
A \emph{bottleneck matching} is, among all possible matchings, a matching with minimal diameter. 
The \emph{diameter} $d$ of $(C_s,C_t)$ is the diameter of a bottleneck matching. 
The \emph{stretch (factor)} of a (stable) schedule is the ratio between its makespan $M$ and the diameter $d$ of~$(C_s,C_t)$.

\smallskip
We consider the \textsc{Connected Coordinated Motion Planning Problem} that is stated as follows; see~\cref{fig:example} for an illustration. 
Given a pair $(C_s, C_t)$ of unlabeled connected start and target configurations, and an integer $k$, we are ask to decide whether there is a stable schedule with a makespan of $k$ that transforms $C_s$ into $C_t$.

%% file: 03-hardness-full.tex
\section{Makespan 1 and 2}\label{sec:NPhigh}

As a first observation we note that it can be decided in polynomial time whether there is a schedule $C_s \rightarrow C_t$ with a makespan of 1 between a start and a target configuration. 

\begin{theorem}\label{lem:single-move-reconfiguration}
	For a pair of configurations $C_s$ and $C_t$, each with $n$ vertices, it can be decided in polynomial time whether there is a schedule with a makespan of $1$ transforming $C_s$ into $C_t$.
\end{theorem}

\begin{proof}
	Given two connected configurations $C_s$ and $C_t$, each with $n$ vertices. 
	We compute the bipartite graph $G_{C_s,C_t}=(V_s \cup V_t, E)$, where $V_s$ and $V_t$ consist of all occupied positions in $C_s$ and $C_t$. 
	For $E$, we add an edge if and only if an occupied position in $C_t$ is adjacent (or identical) to an occupied position in $C_s$. 
	
	Consider a perfect matching in $G_{C_s,C_t}$. 
	By construction, the edges in $G_{C_s,C_t}$ only connect positions from start and target configuration that are at most one unit step apart. 
	As there is a perfect matching, no two robots want to occupy the same position. 
	Furthermore, two paths of length $1$ can only cross at a common vertex; this cannot happen, as we consider a perfect matching; therefore, all robots can move along there respective matching edges in parallel without collisions.
	
	If there is no perfect matching in $G_{C_s,C_t}$, at least one robot would have to move to a position further away. 
	Thus, a makespan of 1 would not be achievable. 
	So, there is a schedule of makespan 1 if and only if $G_{C_s,C_t}$ admits a perfect matching. 
	Because the graph is sparse, this can be checked in $\mathcal{O}(n^{1.5})$ time, using the method of Hopcroft and Karp~\cite{hk-ammbg-73}.
\end{proof}

Note that, because $C_s$ and $C_t$ have to be connected, a schedule with a makespan of~$1$ is always stable. 
Furthermore, it is easy to see that the method described in the proof of~\cref{lem:single-move-reconfiguration} can be applied iteratively to verify whether a suggested schedule with makespan~$k\in \mathbb{N}$ is stable, and indeed transforms the given start configuration in the respective goal configuration. 
Hence, we conclude that the problem is contained in \NP.

\begin{corollary}
	For any pair of configurations, any given (stable) schedule with makespan $k\in \mathbb{N}$ can be verified in polynomial time, i.e., this problem is in \NP.
\end{corollary}

However, even for a makespan of~$2$, the same problem becomes provably difficult.
We show the following theorem.

\begin{restatable}{theorem}{complexitytheorem}\label{thm:connected-motion-planning-hard}
	For a pair of configurations $C_s$ and $C_t$, each with $n$ vertices,
	deciding whether there is a stable schedule with a makespan of $2$ transforming $C_s$ into $C_t$ is \NP-complete.
\end{restatable}

The proof is based on a reduction from the \NP-hard problem \textsc{Planar Monotone 3Sat}~\cite{dbk-obspp-10}, which asks to decide whether a Boolean 3-CNF formula~$\varphi$ is satisfiable, for which in each clause the literals are either all unnegated or all negated. 
Note that a general instance of \textsc{Planar Monotone 3Sat} may contain clauses with three literals as well as clauses with only two literals.

For the following, we refer to~\cref{fig:hardness-reduction-sketch}. The reduction works as follows: For~every instance $\varphi$ of \textsc{Planar Monotone 3Sat}, we construct an instance~$I_\varphi$, consisting of a start configuration $C_s$ and a target configuration $C_t$. In the figure, we use three different colors to indicate occupied positions in the start configuration (red), in the target configuration (dark cyan), and in both configurations (gray). Therefore, 
we consider a rectilinear planar embedding 
of the variable-clause incidence graph $G_{\varphi}$ 
of ${\varphi}$ where the variable vertices are placed horizontally in row, and clauses containing unnegated and negated literals are placed above and below, respectively. All variables
of ${\varphi}$ are represented by a horizontal \emph{variable gadget} (light red).
Furthermore, we position two additional \emph{auxiliary gadgets} (light blue) at the top and
at the bottom boundary of the instance, which are connected to the variable
gadget via bridges at the right boundary.
There will be a
\emph{separation gadget} (yellow) between each adjacent and nested pair of \emph{clause
	gadgets} (blue). All clause gadgets are connected via bridges to separation gadgets
and possibly to the auxiliary gadgets. Further, there are bridges from a
clause gadget to the respectively contained variables. 

\begin{figure}[p]
	\centering
	\includegraphics[width=\linewidth, page=1]{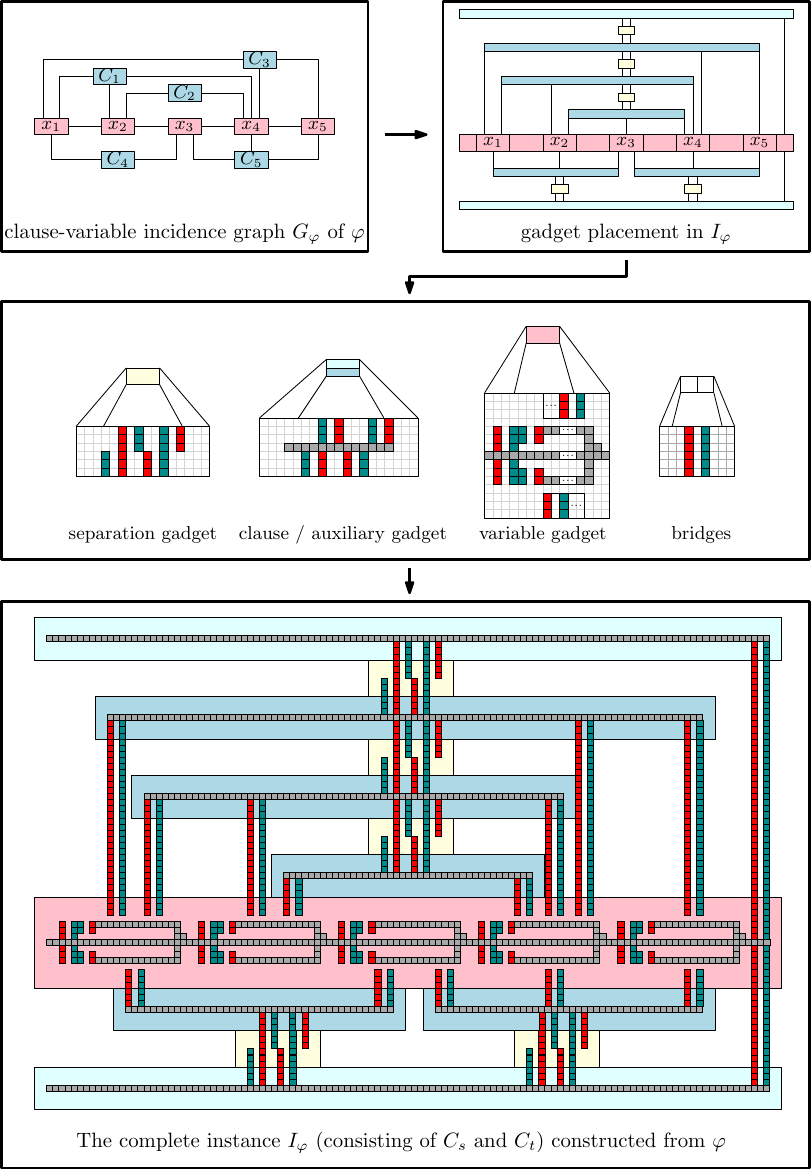}
	\caption{Symbolic overview of the \NP-hardness reduction. The depicted instance is due to the \textsc{Planar Monotone 3Sat} formula
		$\varphi = (x_1 \vee x_2 \vee x_4) \wedge (x_2 \vee x_4) \wedge (x_1
		\vee x_4 \vee x_5) \wedge
		(\overline{x_1} \vee \overline{x_3}) \wedge (\overline{x_3}
		\vee \overline{x_4} \vee \overline{x_5})$. We use three different colors to indicate occupied positions in the start configuration (red), in the target configuration (dark cyan), and in both configurations (gray).}
	\label{fig:hardness-reduction-sketch}
\end{figure}

There is a stable schedule for $I_{\varphi}$ transforming the start configuration $C_s$ into the target configuration $C_t$ with a makespan of $2$, if and only if $\varphi$ is satisfiable. 
In order to transform $C_s$ into $C_t$, the separation gadgets ensure that in the single intermediate configuration, all clause and auxiliary gadgets are disconnected from each other. 
Therefore, to satisfy the connectivity constraint, some robots of the variable gadget need to move to provide connectivity between the variable gadget and the clause gadgets. 
At the same time, our construction ensures that robots representing a variable can either keep connectivity to their unnegated or to their negated literal containing clauses within a makespan of $2$, because otherwise the connectivity within the variable gadget would be broken; thus, these movements can be used to determine a valid variable assignment for $\varphi$.

\subsection{Construction of the Gadgets}\label{sec:gadget-construction}

We start with some observations regarding the target neighborhood of a given
robot, i.e., the positions that can be reached by moving a robot twice from a
given start position. 

\begin{observation}\label{obs:target-neighborhood}
	For a given position, the set of positions reachable within two moves is given by \cref{fig:target-neighorhood-bridges-a}.
\end{observation}

Because we can swap the roles of the start and the target
configurations, an easy consequence is that if there is only one start position in the reachable
neighborhood of a target position, then the robot occupying this position in
the target configuration is uniquely~defined.

\begin{figure}[htb]
	\centering
	\begin{subfigure}{.2\textwidth}
		\includegraphics[scale=0.5, page=2]{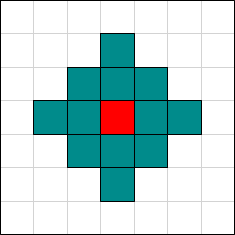}
		\caption{}
		\label{fig:target-neighorhood-bridges-a}
	\end{subfigure}\hfil
	\begin{subfigure}{.2\textwidth}
		\includegraphics[scale=0.5, page=3]{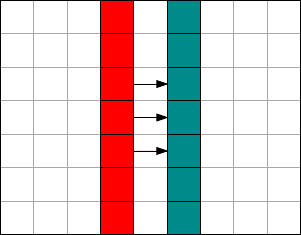}
		\caption{}
		\label{fig:target-neighorhood-bridges-b}
	\end{subfigure}\hfil
	\caption{(a) All positions that can be reached by moving the robot initially located on 
		the red square twice. (b) This figure shows a bridge. Note that all white squares have to be free
		in $C_s$ as well as in $C_t$.} 
	\label{fig:target-neighorhood-bridges}
\end{figure}

In our construction, we place gadgets in such a way that they are ``sufficiently far apart''.
The reason for this is that robots that belong to a certain gadget only can occupy target positions within the same gadget. 
In particular, we place our gadgets in such a way that a robot belonging to some gadget would need at least \num{6} steps to reach another gadget.
This allows us to make statements about the movement within gadgets, as well as their interaction with others. 
We refer to all positions that belong to a gadget as $\mathcal{G}$, and to the positions between gadgets as~$\mathcal{B}$.
In order for the gadgets to interact with each other, we place robots in $\mathcal{B}$ (the bridges) in such a way that their respective target positions are uniquely defined in~$\mathcal{B}$.
With this, it is easy to see that gadgets are indeed independent from each other.
We call a gadget \emph{solvable}, if there is a schedule that transforms the start configuration into the target configuration with a makespan of $2$, and each robot in~$\mathcal{G}$ is connected to a robot in~$\mathcal{B}$. 
Note that we do not necessarily require the intermediate configuration to be connected in order to say that a gadget is solvable; this will not be the case in the separation gadget.
Consider any configuration $C$ that is composed of gadgets constructed in
this manner, 
then a necessary condition for the existence
of a schedule for $C$ with a makespan of $2$ is that every gadget is~solvable. 

In the following, we introduce different gadgets that are used in the proof 
showing \NP-hardness of the problem. We show that these gadgets are solvable, and we indicate explicit~schedules.

\paragraph*{Line Gadget.} A \emph{line gadget} is a rectangular region of size $\ell \times 7$. 
It consists of $\ell - 4\geq 3$ horizontally adjacent robots that occupy the same positions in both, the start and the target configuration, as well as a variable number of bridges of length three, as depicted in~\cref{fig:line-gadget}.

\begin{figure}[htb]
	\centering
	\includegraphics[scale = 0.5]{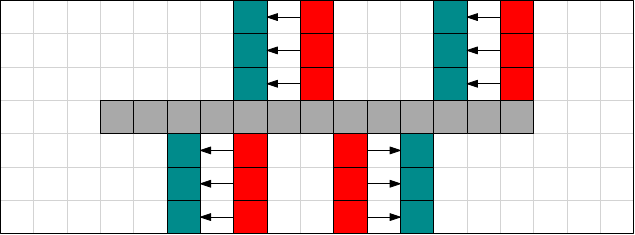}
	\caption{The gray positions indicate the line gadget. The colored positions indicate the start and target configurations of the respective bridges.}
	\label{fig:line-gadget}	
\end{figure}

In the \NP-hardness proof, line gadgets are used in two ways:
As \emph{clause gadgets}, they represent the logic of the satisfiability instance;
as \emph{auxiliary gadgets}, they ensure connectivity of the start and the
target configuration, respectively.

\begin{lemma}\label{lem:motion-plan-line}
	The line gadget is solvable.
\end{lemma}

\begin{proof}
	All robots realizing bridges move two steps to the right or to the left to their respective target positions. The line robots do not move at all. Because no robot moves vertically, all configurations are connected.
\end{proof}

\paragraph*{Separation Gadget.} A \emph{separation gadget} is a rectangular region of size $14\times 6$, containing the start and the target configuration as shown in~\cref{fig:hardness-separation-gadget-a}.

\begin{figure}[htb]
	\centering
	\begin{subfigure}{0.3\textwidth}
		\includegraphics[scale = 0.5, page=2]{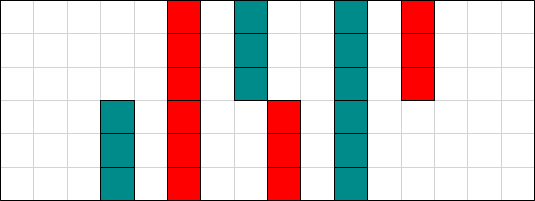}
		\caption{}
		\label{fig:hardness-separation-gadget-a}
	\end{subfigure}\hfil
	\begin{subfigure}{0.3\textwidth}
		\includegraphics[scale = 0.5, page=3]{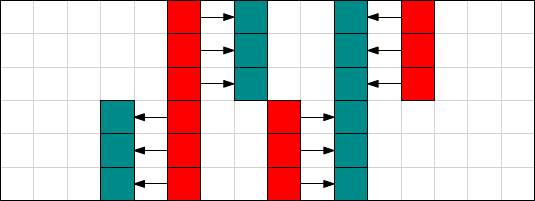}
		\caption{}
		\label{fig:hardness-separation-gadget-b}
	\end{subfigure}\hfil
	\caption{(a) shows the separation gadget. (b) visualizes the unique schedule with a makespan of~2.}
	\label{fig:hardness-separation-gadget}	
\end{figure}

\begin{lemma}\label{lem:motion-plan-separation}
	The separation gadget is solvable.
\end{lemma}

\begin{proof}
	There is a unique schedule with a makespan of 2, as shown in
	\cref{fig:hardness-separation-gadget-b}. This follows from
	applying \cref{obs:target-neighborhood} to the different target
	positions.  
\end{proof}
\clearpage

In the intermediate configuration of the unique schedule solving a separation
gadget, each pair of bridges is not connected within the gadget, i.e., the intermediate configuration is not connected. For that reason, these gadgets will be used to disconnect the connections between all clause gadgets. To guarantee a fully connected intermediate configuration, a specific movement within the variable gadget is necessary.

\paragraph*{Variable Gadget.} Let $n$ be the number of variables of the {\sc Planar Monotone 3Sat} instance. 
The \emph{variable gadget} is composed of the parts shown in~\cref{fig:variable-gadget}, as follows: 
There is one \emph{left end} followed by $n$ horizontally aligned \emph{variable arm segments} and one \emph{right end}. 
In order to capture the respective number of clauses in which a variable is contained, we can independently adjust the width of each arm segment.  
Bridges are placed in the gray colored hatched parts. 
We refer to the top and the bottom parts of the arm segments as the \emph{unnegated} and the \emph{negated arm segment}, as they model the unnegated and negated literals, respectively. 
As an example, consider~\cref{fig:hardness-variable-gadget-example}.

\begin{figure}[p]
	\centering
	\includegraphics[scale = 0.5]{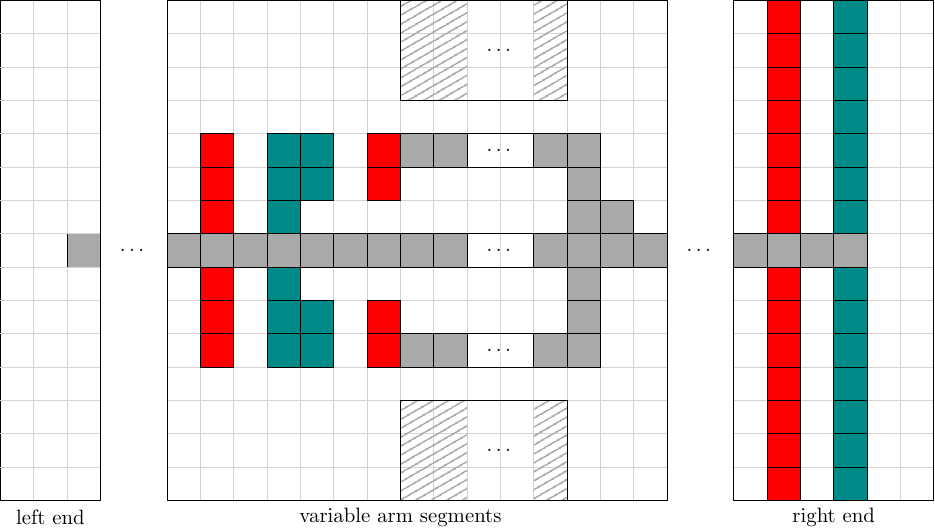}
	\caption{The variable gadget consists of a number $n$ variable arm segments, and exactly one left and right end each. Bridges to the respective clauses are placed in the hatched areas.}
	\label{fig:variable-gadget}	
\end{figure}

\begin{figure}[p]
	\centering
	\includegraphics[scale = 0.8]{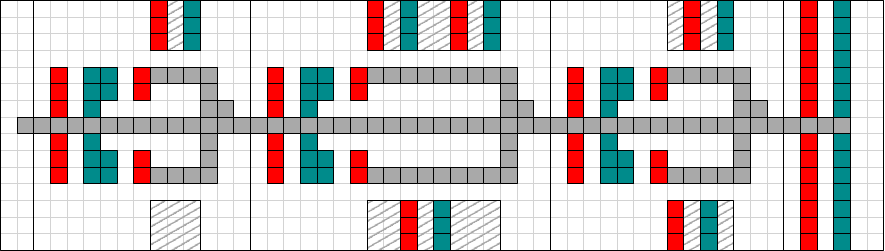}
	\caption{An example construction of the variable gadget with three variables. Note that each arm segment is of different width, due to the number of clauses in which the variable is contained.}
	\label{fig:hardness-variable-gadget-example}	
\end{figure}

\begin{figure}[p]
	\centering
	\includegraphics[scale = 0.8]{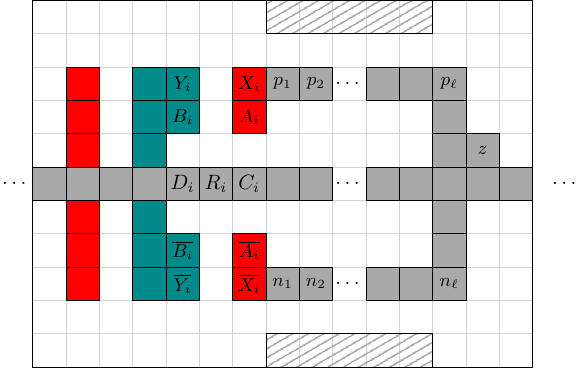}
	\caption{The figure shows the labeling of specific robots of a variable arm segment for the proof of \cref{claim:variable-one-literal}. The positions $p_i$ and $n_i$ denote the unnegated and negated arm segment, respectively.}
	\label{fig:variable-gadget-names}	
\end{figure}

\begin{lemma}\label{claim:variable-one-literal}
	In the variable gadget, it is not possible that both, the unnegated and the negated arm segment of a variable arm segment
	are simultaneously connected to its respective bridges.
\end{lemma}

\begin{proof}
	Without loss of generality, we assume that in the intermediate configuration
	the robots representing the unnegated arm segment are connected to the
	respective bridges. We~use the labels shown in
	\cref{fig:variable-gadget-names}. We argue that the following properties
	hold after the first~transformation.
	\smallskip 
	\begin{itemize}
		\item At least one robot of the unnegated arm segment ($p_1, \dots, p_\ell$) moved up.
		\item Robot $A_i$ moved left. 
		\item Robot $R_i$ moved up.
	\end{itemize}
	\smallskip
	Assume for the sake of contradiction that no robot of the unnegated arm segment
	moved up. Because the top two bridge robots moved horizontally (by applying
	\cref{obs:target-neighborhood}), every movement of the bridge's
	bottommost robot would not yield a connection to the arm~segment.
	
	The robot $A_i$ has two possible target positions, namely $B_i$ and
	$C_i$. Assume that $A_i$ moved down. Then, $X_i$ has the unique target position
	$Y_i$, so it moved left---thus, it is isolated in the intermediate
	configuration. Therefore, $p_1$ has also moved left. It follows by induction
	that every robot of the unnegated arm segment has moved left, which is a
	contradiction to the fact that at least one robot moved up. Thus, $A_i$
	moved left.
	
	Assume that $R_i$ did not move up. This is a contradiction to $A_i$ having moved left, as otherwise at least $A_i$ would be part of an isolated component in the intermediate configuration.
	
	Similar arguments hold for the case that the robots of the
	negated arm segment are connected to its respective bridges. Therefore, $R_i$
	has to move both up \emph{and} down in the first transformation. As this is
	impossible, this concludes the~proof.  
\end{proof}

\begin{lemma}\label{lem:motion-plan-variable}
	A variable gadget is solvable.
\end{lemma}

\begin{proof}
	There is a valid schedule with a makespan of 2 given in
	\cref{fig:variable-gadget-schedule}. The depicted schedule connects the
	unnegated arm segment with the respective bridges; a similar schedule exists for the negated case.  
\end{proof}

\begin{figure}[htb]
	\centering
	\begin{subfigure}[b]{0.25\textwidth}
		\includegraphics[page=2, scale=0.8]{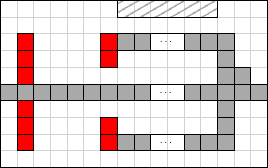}
		\caption{}
		\label{fig:variable-gadget-schedule-a}
	\end{subfigure}\hfil
	\begin{subfigure}[b]{0.25\textwidth}
		\includegraphics[page=3, scale=0.8]{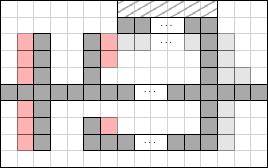}
		\caption{}
		\label{fig:variable-gadget-schedule-b}
	\end{subfigure}\hfil
		\begin{subfigure}[b]{0.25\textwidth}
		\includegraphics[page=4, scale=0.8]{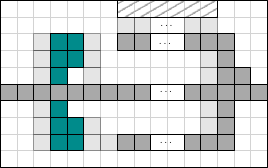}
		\caption{}
		\label{fig:variable-gadget-schedule-c}
	\end{subfigure}\hfil
	\caption{The figure shows a stable schedule with a makespan of 2 as indicated in \cref{lem:motion-plan-variable}. Positions with less opacity show the respectively occupied positions in the previous step. (a) shows that start configuration, while (b) indicates the single intermediate configuration, and (c) depicts the respective target configuration.}
	\label{fig:variable-gadget-schedule}	
\end{figure}

\subsection{Completing the NP-Hardness Reduction}

\complexitytheorem*

\begin{proof}
	Consider a rectilinear planar embedding of the
	variable-clause incidence graph $G_\varphi$ of a given {\sc Planar Monotone 3Sat} formula $\varphi$. The horizontally aligned
	variables in $G_\varphi$ are represented by a variable gadget---each variable is
	represented by an arm segment. For each clause, we introduce a line gadget with
	connecting bridges (and extend them if necessary) to the contained variables. 
	We add two or three bridges according to whether the clause contains two or three literals, respectively.	
	Between every nested pair of adjacent clauses, we introduce a separation gadget
	and connect 
	it to the clauses via bridges. As~a last step, we place two additional line gadgets
	at the top and bottom of the construction,  
	connect them via bridges to the
	variable gadget, and via separation gadgets to the respective topmost and
	bottommost line gadgets that represent clauses. For an example, see the bottom part of~\cref{fig:hardness-reduction-sketch}.
	
	\begin{claim}\label{claim:hardness-one}
		If the formula $\varphi$ has a satisfying assignment $\alpha$,
		then there is a valid intermediate configuration for $I_{\varphi}$ (i.e., the intermediate configuration is connected); therefore,
		there is a stable schedule transforming $C_s$ into $C_t$ (both given by $I_{\varphi}$) with a makespan of 2.  
	\end{claim}
	
	\begin{claimproof}
		Let $\alpha$ be a satisfying assignment of $\varphi$. A valid
		intermediate configuration for $I_{\varphi}$ can be constructed as follows: All
		separation, line, and variable gadgets are transformed by the schedules given
		in \cref{lem:motion-plan-line,lem:motion-plan-separation,lem:motion-plan-variable}. For the assignment of variable $x_i$ the
		respective arm segment moves vertically. All bridge robots move to their
		respective target position due to \cref{obs:target-neighborhood}. In
		this configuration the helping line gadgets are directly connected to the
		variable gadget, and all clause gadgets are connected to the variable gadget as
		well. Because the disjoint parts of the separation gadgets are connected to a
		line gadget at one side, this configuration is connected.
	\end{claimproof}
	
	\begin{claim}\label{claim:hardness-two}
		If there is a connected intermediate configuration for $I_{\varphi}$ (such that there is a stable schedule transforming $C_s$ into $C_t$ (both given by $I_{\varphi}$) with a makespan of 2), then $\varphi$ is~satisfiable.
	\end{claim}
	
	\begin{claimproof}
		Because the intermediate configuration is connected, each
		clause gadget has to be connected to the variable gadget. This is only possible
		via the bridges at the boundary of the gadget.  
		Due to the construction of
		the separation gadgets, the clause gadgets have to be \emph{directly} connected
		to the variable gadget, i.e., a shortest path connecting a robot of each clause
		gadget with a robot of the variable gadget passes no other gadget of the
		construction. Thus, at least one of the three bridges of each clause is
		connected to a respective arm segment of the variable gadget. Because of
		\cref{claim:variable-one-literal}, for each variable either its unnegated
		or its negated arm segment can be connected to their respective bridges. Therefore, there is an
		assignment of $\{1,0,\bot\}$ for $\varphi$, given by the movement of each variable arm segment. We set a variable in $\varphi$ to $1$, if the respective arm segment is connected to the unnegated side, $0$ if it is connected to its negated side, and $\bot$ otherwise. 
		For each $\bot$ we arbitrarily choose $0$ or
		$1$. Due to the construction, this assignment satisfies $\varphi$.
	\end{claimproof}
	\cref{claim:hardness-one,claim:hardness-two} complete the proof of \cref{thm:connected-motion-planning-hard}.
\end{proof}

As a consequence of our construction in the proof of  \cref{thm:connected-motion-planning-hard}, even approximating the makespan is \NP-hard.

\begin{corollary}\label{cor:connected-motion-planning-optmial-hard}
	It is \NP-hard to compute for a pair of configurations $C_s$ and $C_t$, each with $n$ vertices,
	a stable schedule that transforms $C_s$ into $C_t$ 
	within a constant of $(\frac{3}{2}-\varepsilon)$ (for any $\varepsilon>0$) of the minimum makespan.
\end{corollary}

To see this, consider an instance that derived from a Boolean formula as constructed in our \NP-hardness reduction, and assume that the Boolean formula is satisfiable. 
Hence, the optimal makespan is \num{2}. 
Consider a potential approximation algorithm that has a guaranteed approximation factor better than $\frac{3}{2}$. 
As~the makespan is a natural number, and the minimum makespan of a false instance is \num{3}, we would be able to decide with that algorithm whether the Boolean formula is satisfiable, having as a consequence that \P=\NP.

%% file: 04-algorithm-full.tex
\section{Bounded Stretch for Arbitrary Makespan}

Now we describe our algorithm for computing stable schedules with constant stretch, for configurations of constant scale. 
Again, a scale factor of $c$ corresponds to replacing each pixel of a polyomino shape by a quadratic $c\times c$ array of pixels; this will be defined in the preliminaries for the algorithm. 
In the remainder of this section, we describe the different phases of our approach that together show the following.

\begin{theorem}\label{thm:scaled-instances-bounded-stretch}
	There is a constant $c^*$ such that for any pair of overlapping start and target configurations with a scale of at least $c^*$, there is a stable schedule of constant stretch.
\end{theorem}

For clearer presentation, we do not focus on the specific value of the constant~$c^*$, but only argue its existence.

\subsection{Algorithm Overview and Preliminaries}

We provide a high-level overview of our algorithm, as depicted in~\cref{fig:overview}.
\begin{figure}[ht]
	\centering
	\includegraphics[width=0.9\textwidth]{graphics/tile_tiled_conf_new3}
	\caption{Overview of the computed schedule: (Top) Constructing an auxiliary structure, called \emph{scaffold}, in Phase~3 of our algorithm, (right) the refilling phase (Phase~4), and (bottom) deconstructing the scaffold in the final Phase~5.}
	\label{fig:overview}	
\end{figure}

Our algorithm works in five phase, where the first two phases are preparation phases; note that these two phases are not visualized in~\cref{fig:overview}.
In Phase~1, we ensure that the pair $(C_s, C_t)$ overlaps in at least one position. 
For this, we move $C_s$ towards $C_t$ along a bottleneck matching, such that the respective positions that realize the bottleneck distance, coincide. 
(This overlap is crucial for successfully constructing the auxiliary structure in Phase~3 of the approach, as we explain later.) 
Afterwards, in Phase~2, we use another bottleneck matching for mapping the start configuration $C_s$ to the target configuration $C_t$, minimizing the maximum distance $d$ between a start and a target location.
Furthermore, we establish the scale in both configurations, set $c$ to be the minimum of both scale values, and compute a suitable tiling whose tile size is $c\cdot d$, and that contain both $C_s$ and~$C_t$.

In Phase~3, we build a scaffolding structure around $C_s$ and $C_t$, based on the boundaries of \emph{$cd$-tiles} of the specific tiling, see~\cref{fig:overview,fig:scaffolding}. 
This structure enables simple connected reconfiguration procedures in the interior of tiles, as described in Phase~4 of our approach, yielding to a reconfiguration approach that consists of a connected configuration in each transformation step. 
Restricting robot motion to their current and adjacent tiles (due to the tile size) also ensures constant stretch. 
Note that, as the size of the tiles is related to $d$, the scaffolding structure is connected.

In Phase~4, we perform the actual reconfiguration of the arrangement. 
This consists of refilling the tiles of the scaffold structure, achieving the proper number of robots within each tile, based on elementary flow computations. 
As a subroutine, we transform the robots inside each tile into a canonical ``triangle'' configuration, see~\cref{fig:overview,fig:triangular-shape,fig:refilling_overview}. 

In the final Phase~5, we disassemble the scaffolding structure and move the involved robots to their proper destinations, see~\cref{fig:overview,fig:scaffolding}.

We proceed by providing the basic definitions, followed by summarizing the technical key components of our approach.

\subsubsection{Preliminaries for the Algorithm}
\label{sec:algprelim}

A configuration $C$ is \emph{$c$-scaled}, if it is the union of $c\times c$ pairwise disjoint squares of vertices. 
The \emph{scale} of a configuration~$C$ is the maximal $c$ such that $C$ is $c$-scaled; this notion reflects the idea of objects being composed of pixels at a certain resolution. 
(Note that this is a generalization of the uniform pixel scaling studied in previous literature, which considers a $c$-grid-based partition instead of an arbitrary union, so it supersedes that definition and leads to a more general set of results; for simplicity, we refrain from using ``generalized scale'' and stick to ``scale'' in the rest of the paper.)
As we will establish, sufficient scale ensures sufficient locally available building material for connected rearrangement.
Let $c, d \in \mathbb{N}$ be the scale and the diameter of the pair $(C_s, C_t)$, respectively. 
For $x,y \in \mathbb{N}$, a \emph{$cd$-tile $T$}, or
\emph{tile~$T$} for short, with \emph{anchor vertex} $(x\cdot cd, y \cdot cd)
\in V(G)$ is a set of $(cd)^2$ vertices from the grid $G$ with $x$-coordinates
from the range between $x\cdot cd$ and $x\cdot cd + cd - 1$ and $y$-coordinates
from the range between $y\cdot cd$ and $y\cdot cd + cd - 1$.
Note that, by definition, two tiles are pairwise disjoint, and all anchor points have the same coordinates $\bmod\ cd$. 
The \emph{boundary} of $T$ is the set of vertices from $T$ with an $x$-coordinate
equal to $x\cdot cd$ or equal to $x \cdot cd + cd - 1$, or with a $y$-coordinate
equal to $y\cdot cd$ or equal to $y \cdot cd + cd - 1$. The~\emph{interior} of
$T$ is $T$ without its boundary. The \emph{right}, \emph{top}, \emph{left}, and
\emph{bottom sides} of $T$ are the sets of vertices from the boundary of $T$
with maximum $x$-coordinates, maximum $y$-coordinates, minimum
$x$-coordinates, and minimum $y$-coordinates, respectively. The left and right sides of a
tile are \emph{vertical sides} and the top and bottom sides are
\emph{horizontal sides}. Two tiles $T_1,T_2$ are \emph{horizontal (vertical)
	neighbors} if they have two horizontal (vertical) sides $s_1 \subset T_1$ and
$s_2 \subset T_2$, such that each vertex from $s_1$ is adjacent in $G$ to a
vertex from $s_2$. Two tiles $T_1$ and $T_2$ are \emph{diagonal neighbors} if
there is another tile $T$, such that $T$ and $T_1$ are horizontal neighbors and
$T$ and $T_2$ are vertical neighbors. The \emph{neighborhood} of a tile $T$ is
the set of all neighbors of $T$. 

A \emph{start tile} is a tile containing a vertex from the start configuration. A \emph{target tile} is a tile containing a vertex from the target configuration. The \emph{$cd$-tiling} $\mathcal{T}$ of $(C_s,C_t)$ is the union of all start tiles including their neighborhoods and all target tiles. The \emph{scaffold} $\Sigma(\mathcal{T})$ is the union of all boundaries of tiles from $\mathcal{T}$. A~\emph{$cd$-tiled configuration}, or \emph{tiled configuration} for short, is a configuration that is a subset of the union of all tiles from $\mathcal{T}$ and a superset of~$\Sigma(\mathcal{T})$. The \emph{interior} of a tiled configuration $C$ is the set of all vertices from $C$ not lying on~$\Sigma(\mathcal{T})$. The \emph{filling level of a tile} $T \in \mathcal{T}$ is the number of robots in the interior of $T$. The \emph{filling level of a tiled configuration} $C$ is the mapping of each tile onto its filling level in $C$.

\subsubsection{Technical Key Components}

On a technical level, the five phases can be summarized as follows. 
\medskip
\begin{description}
	\item[(1) Guaranteeing Overlap.] Move the given configurations towards each other along a bottleneck matching to ensure that the (new) pair $(C_s,C_t)$ overlaps in at least one position.
	\item[(2) Preprocessing.] Apply the following three steps: (2.1)~Set~$c$ to be the minimum of $c^*$ (that is a roughly estimated constant lower bound on the scale for which our approach works) and the minimum scale values of $C_s$ and $C_t$. (2.2)~Compute the diameter $d$ of~$(C_s,C_t)$. (2.3)~Compute the tiling $\mathcal{T}$ of~$(C_s,C_t)$.
\end{description}

The algorithmic core of our algorithm consists of the following three phases.
\smallskip
\begin{description}
	\item[(3) Scaffold Construction.] Reconfigure the start configuration $C_s$ to a tiled configuration~$C_s'$ such that the interior of $C_s'$ is a subset of the start configuration~$C_s$, see \cref{fig:overview,fig:scaffolding}.
	\item[(4) Refilling Tiles.] Reconfigure $C_s' $ to a tiled configuration $C_t'$, such that the interior of $C_t'$ is a subset of the target configuration $C_t$, see \cref{fig:overview,fig:triangular-shape,fig:refilling_overview}.
	\item[(5) Scaffold Deconstruction.] Reconfigure $C_t'$ to $C_t$, see \cref{fig:overview,fig:scaffolding}.
\end{description}

Note that the scaffold deconstruction is inverse to the scaffold construction. 
In the following, we give the technical description of our algorithm, and the corresponding correctness analysis. In particular, we first assume that the start and target configurations overlap in at least one position, resulting in an algorithm guaranteeing constant stretch, and adapt this to the case in that an overlap initially does not exist, afterwards.

\subsection{Scaffold Construction}\label{sec:tiled-configurations}

\begin{lemma}
	For any configuration $C_s$ of scale $c$ there is a stable schedule of makespan~$\mathcal{O}(d)$, transforming $C_s$ into a tiled configuration $C'_s$, with the interior of $C'_s$ being a subset of $C_s$.
\end{lemma}

\subparagraph*{Outline of the Construction.} 
For the construction, we consider $5 \cdot 5$ different classes, based on
$x$- and $y$-coordinates modulo $5cd$; see \cref{fig:scaffolding}. 
\begin{figure}[htb]
	\centering
	\includegraphics[width=\textwidth]{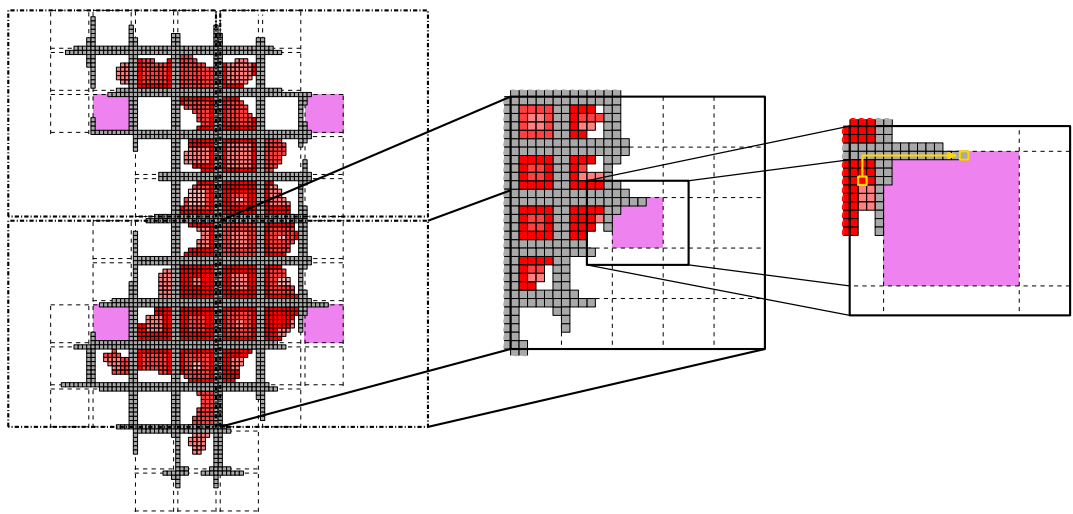}
	\caption{Constructing the scaffold. Tiles with currently constructed boundaries are marked in purple. The zoom into the start configuration 
		$C_s$ shows the indirect neighborhood $N[T]$ of  
		a tile $T$ (middle) 
		for which its boundary is currently constructed. A further zoom~(right) into $T$ with an associated robot motion.
		In each step a robot from the interior of a tile $T' \in N[T]$ 
		is swapped with a free position on the boundary of $T$ based on a path $P$ on a BFS-tree.}
	\label{fig:scaffolding}
\end{figure}

We process a single class as follows. 
For each tile $T$ we consider its \emph{indirect} neighborhood~$N[T]$ consisting of all neighbors of $T$ and all neighbors of neighbors of $T$, i.e., a $5 \times 5$ arrangement of tiles centered at $T$. For constructing the boundary of $T$, we make use of robots from the interior of a tile in $N[T]$.

\subparagraph*{Constructing the Boundary of $T$ Works in Two Phases.} (3.1) Constructing the boundaries of all start tiles. (3.2) Constructing the boundaries of all neighbors of start tiles.

Note that it suffices to construct all boundaries of the start tiles and their neighboring tiles, because each target tile shares a side with a start tile or a side with a tile adjacent to a start tile. Furthermore, the scale condition is only necessary for the construction of the scaffold, i.e., we need to ensure that enough robots are available to build the scaffolding structure. Each additional step of the algorithm works independently of this condition.

For Phase~3.1, we subdivide all start tiles into $5 \cdot 5$
classes of tiles, distinguishing between their anchors' 
coordinates
modulo $5cd$; see the purple squares in~\cref{fig:scaffolding}. These classes are processed iteratively. In particular, we
construct the boundaries of all tiles belonging to a specific class in parallel
as follows. Without loss of generality, we assume that the union of $C_s$ with all boundaries of
start tiles is connected. Otherwise, $C_s$ lies completely inside a single tile
$T$, and we move all robots simultaneously down until the first robot lies
on the boundary of $T$.

For each tile $T$ of the current class, we consider a $5 \times
5$-neighborhood $N[T]$ of $T$ which consists of $T$, its neighbors, and all
neighbors of neighbors of $T$, see~\cref{fig:scaffolding}. The~scaffold~$\Sigma(N[T])$ is the union of all boundaries of tiles from~$N[T]$. The interior of~$N[T]$ is~$N[T]$ without its scaffold, i.e., $N[T]\setminus \Sigma(N[T])$. 
The \emph{priority} of a robot $r$ from
the interior of $N[T]$ is the shortest-path distance inside the current
configuration between $r$ and any vertex from~$\Sigma(N[T])$.

Let $C$ be the current configuration, $v$ be a vertex from the boundary of
$T$ that does not lie inside $C$ but is adjacent to a robot from $C$, and $r$ be a robot
from the interior of $N[T]$ with highest priority. If  
$C$ is connected and 
the robot~$r$ lies in the interior of $N[T]$, then there is a path $P$ inside
$C$ connecting $v$ and $r$, see the right of \cref{fig:scaffolding}. We
push all robots on $P$ along $P$ into the direction of~$v$ resulting in the
connected configuration $C \cup \{ v \} \setminus \{ r \}$. This one-step
motion is repeated until the interior of $N[T]$ is empty or the
boundary of $T$ is a subset of the current configuration.

\subparagraph*{Correctness of the Approach.} We need to argue that in Phase~3.1 the construction of the boundary of~$T$ stops because the boundary is completed and not because the interior of $N[T]$ is empty. 

Consider the vertices from $N(T) := N[T] \setminus T$ to be
organized in squares nested inside each other as follows. All vertices $v \in
N(T)$ adjacent to the boundary of $T$ belong to layer~$1$. Let $N_1(T)$ be the
vertices from $N(T)$ without the vertices from layer $1$. All vertices from~$N_i(T)$ that are adjacent to vertices from layer $i$ belong to Layer $i+1$.
Consider a path~$P \subset C_s$ connecting layer~$c$ with layer $cd-c$. As
$C_s$ is $c$-scaled, for each $v \in P$, there is a square $S$ of side length
$c$ such that $v \in S \subset C_s$. Hence, there is a tile $T' \in N[T]$, such
that all these squares cover at least $\frac{(cd - 2(c-2))(c-2)}{2} \geq
\frac{(c(d - 2))(c-2)}{2}$ vertices from $C_s$, which is lower-bounded by
$\frac{c^2d}{4}$ for $c \geq 4$. Constructing the boundary
of $T$ takes at most $4cd - 4$ robots to be taken from the interior of $N[T]$.
The tile $T'$ is used as part of the interior of at most $5 \cdot 5 = 25$
tiles' neighborhoods $N[T]$. Thus, at most $(4cd - 4)25 < 100 cd$ robots are
required from the interior of the tile $T' \in N[T]$. Hence, for
$\frac{c^2d}{4} \geq 100 cd$, i.e., $c \geq 400$  the construction of the
boundary of each start tile $T$ is completed in Phase~3.1. Because there is a constant number of tile classes (where each class can be processed in parallel), and the moving distances are bounded by a constant times $d$, the construction in finishes in~$\mathcal{O}(d)$. Phase~3.2 is realized analogously.

\subsection{Refilling Tiles}
\label{sec:refilling}

It remains to modify configurations within and between tiles. To this end, we first
establish how to efficiently perform reconfigurations between
any two tiled configurations \emph{with the same numbers of robots} {in the interior of respective tiles}; 
see~\cref{sec:reconfig-single-tiles}. {As a second step, we describe
	how to relocate robots between tiles such that efficient reconfigurations
	between any two tiled configurations \emph{with different numbers of robots} in
	the interior of respective tiles are achieved; see
	\cref{sec:realrefilling}.}

\subsubsection{Reconfiguration Maintaining the Number of Robots in Tiles}\label{sec:reconfig-single-tiles}

\begin{lemma}\label{lem:single-tile-reconfig}
	Let $C_s', C_t'$ be two tiled configurations such that for all tiles $T$, $C_s'$ and $C_t'$ have the same filling levels, 
		i.e., for any tile, the corresponding start and target configurations consist of the 
		same respective numbers of robots. Then there is a stable schedule 
	transforming $C_s'$ into $C_t'$ within a makespan of $\mathcal{O}(d)$.
\end{lemma}
In the following we describe reconfigurations that leave all robot movements within the interior of their respective tiles
	$T$; thus, all tiles can be reconfigured in parallel. Therefore, we only have
to describe the approach for a start configuration $C_s$ and a target
configuration~$C_t$ \emph{within the interior of a single tile $T$} of a tiled configuration $C_s'$. 

The idea is as follows: As the moves are reversible, it is sufficient to prove reachability of a canonical intermediate configuration in $\mathcal{O}(d)$ transformations, which allows us to simplify our argument.
In the following, we define the properties of our canonical structure $C_\Delta$ and provide the necessary means to reach it from an arbitrary initial state.

\subparagraph*{Outline of the Reconfiguration.} First, we compute two stable schedules $C_s \rightrightarrows_{\chi} C_s^m$
and $C_t \rightrightarrows_{\chi} C_t^m$, where $C_s^m$ and $C_t^m$ are \emph{monotone
configurations}, where we call a configuration to be monotone, if for every robot $r$ in the interior of $T$ all positions to the left and to the bottom are occupied. These reconfigurations are achieved by a sequence of down and
left movements, maintaining connectivity after each move (see
\cref{fig:triangular-shape}, Phase~4.1). Proceeding from these monotone
configurations, the robots are arranged into a triangular configuration~$C_{\Delta}$ that occupies the lower left positions 
(defined by a diagonal line with a slope of $-1$) of the interior of $T$. 
This is achieved by swapping pairs of occupied and empty positions within a carefully defined area in several one-step moves along
L-shaped paths (see~\cref{fig:triangular-shape}, Phase~4.2). The~property of $C_{\Delta}$ is that it is the
same for all initial configurations with equally many robots. Thus, to get the stable schedule $C_s \rightrightarrows_{\chi} C_{\Delta} \rightrightarrows_{\chi} C_t$, we can simply revert $C_t \rightrightarrows_{\chi} C_{\Delta}$ and combine the result with $C_s \rightrightarrows_{\chi} C_{\Delta}$.

\begin{figure}[ht]
	\centering
	\includegraphics[page = 1, width=\textwidth]{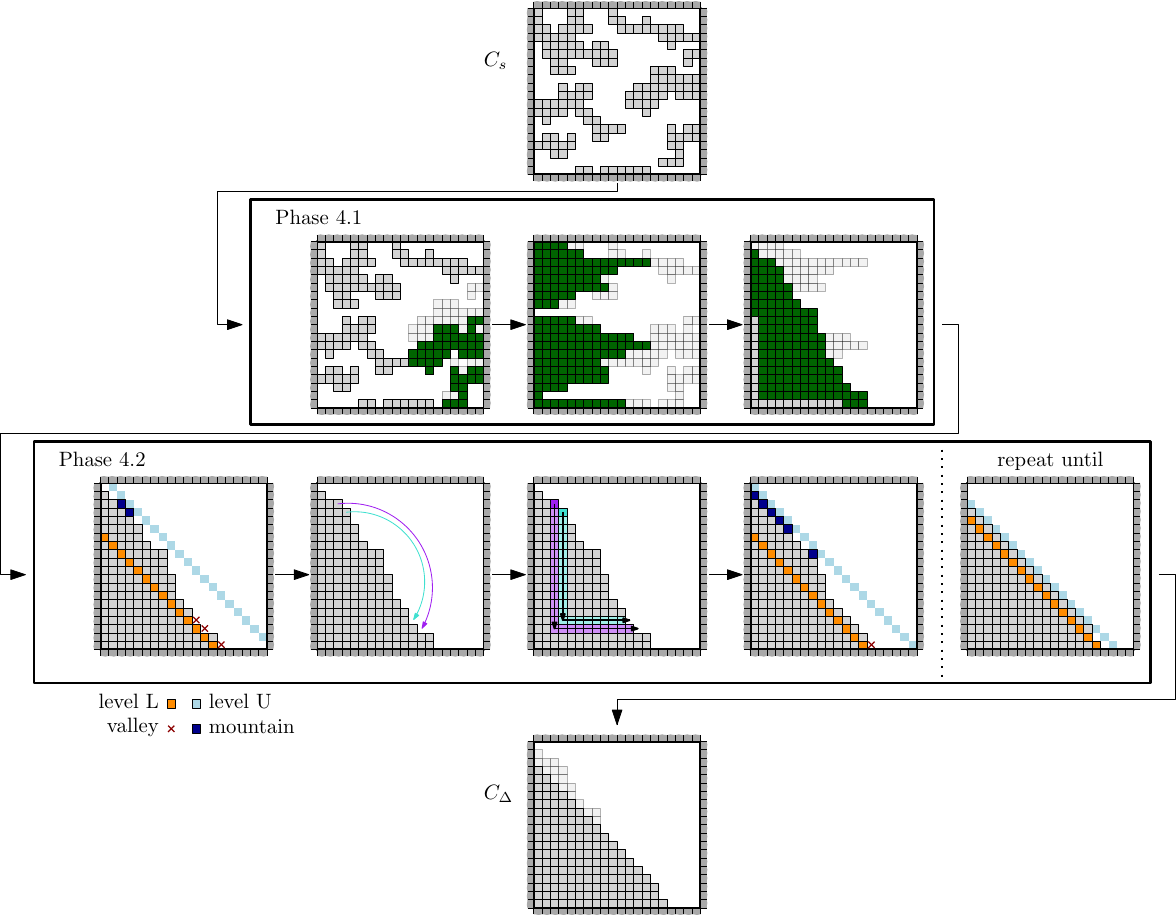}
	\caption{Reconfiguring the configuration $C_s$ (top) into the canonical triangle configuration~$C_\Delta$ (bottom).
		Phase~4.1 achieves a monotonic arrangement; light gray squares indicate previous positions
		of moved robots (shown in green). Phase~4.2 transforms the monotonic configuration into~$C_\Delta$.}
	\label{fig:triangular-shape}
\end{figure}

\medskip
Technically, the approach consists of the following four phases, see \cref{fig:triangular-shape}.
\medskip
\begin{description}
	\item[(4.1) Monotone Start Configuration.] Reconfigure $C_s$ into $C_s^m$.
	\item[(4.2) Canonical Triangle.] Reconfigure $C_s^m$ into $C_{\Delta}$.
	\item[(4.3) Monotone Target Configuration.] Reconfigure $C_{\Delta}$ into $C_t^m$.
	\item[(4.4) Target Configuration.] Reconfigure $C_t^m$ into $C_t$.
\end{description}

Phase~4.4 corresponds to a reversal of Phase~4.1, and Phase~4.3 to one of Phase~4.2, so we only have to describe the first two phases. 
We analyze them individually, leading to a proof of~\cref{lem:single-tile-reconfig}. 
Note that we exclude the corners of a tile, so the robots on the tile's side now form four \emph{non-adjacent} sides. Furthermore, only robots in a tile's interior move.

\subparagraph*{Constructing the Monotone Start Configuration (Phase~4.1).}
In the first step, we only consider robots for which the right side of their tile $T$ is the only one to which they are connected through the interior of $T$.
We iteratively move these robots down until further movement is blocked, i.e., any further down move is not collision-free. 
In the second and third steps, we move all robots left, followed by moving all robots down, each time until further movement is blocked. 

\begin{claim}
	Phase~4.1 turns a configuration in a tiles' interior into a monotone configuration.
\end{claim}

\begin{claimproof}
	It is straightforward to see that for, any configuration, Phase~4.1 results in a monotone configuration, so we only argue that it yields a stable schedule.
	
	Without loss of generality, we only consider the down movement of a
	robot. 
	If the movement of a robot $r$ resulted in a collision, there would be another robot on its target position that cannot move either. 
	Because the robots only move down, it iteratively follows that this movement is blocked because of a robot that is either adjacent to the bottom side or for which the right side is not the only one to which it is connected.
	
	We now argue that a down move of a robot $r$ preserves connectivity to
	other robots. 
	Assume that there is a robot $r'$ to which $r$ was connected before the move (say, in $C_i$), but not after it (in $C_{i+1}$). 
	Consider a path from $r$ to $r'$ in $C_i$, the last robot $r_j$ on that path which is not connected to $r$ in $C_{i+1}$, and its adjacent robot $r_{j+1}$ on that path that is still connected to $r'$; note that this robot could be $r'$ itself. 
	Because $r_j$ and $r_{j+1}$ are not connected in $C_{i+1}$, they were horizontally adjacent in $C_i$. 
	Without loss of generality, this disconnection happened because $r_{j}$ moved down, whereas $r_{j+1}$ did not.
	But then all positions below $r_{j+1}$ had to be occupied. Thus, $r_j$ and $r_{j+1}$ are still connected, a contradiction. 
\end{claimproof}

By simply changing the direction of the robot's movements from down to left, all arguments hold true for the entire phase.

For Phase~4.2, we use the following terminology. 
The \emph{level} of a position in the tile's interior is the sum of its coordinates. 
A level is \emph{filled}, if all of its positions are occupied by a robot, and \emph{empty}, if none is occupied by a robot. 
The highest filled level is denoted by $L$, the lowest empty level by $U$. 
Let $\mountain$ be the set of all positions on level $U-1$ occupied by a robot, and $\valley$ be the set of all positions on level $L+1$ that are not occupied by a robot; we call the positions of $\mountain$ and $\valley$ \emph{mountains} and \emph{valleys}, respectively. 

\subparagraph*{Constructing the Canonical Triangle (Phase~4.2).} Choose two equally sized subsets $\mountain'\subseteq \mountain$ and $\valley'\subseteq \valley$ and push each robot from $\mountain'$ to a different position in $\valley'$ along an L-shaped path; this can be done simultaneously in one parallel move for all paths. To determine the paths, simply match 
mountains and valleys, iteratively, in a way that no pair of paths cross each other. This results in reducing $U$, and raising $L$, i.e., the two levels move towards each other. We distinguish two cases.
\medskip
\begin{itemize}
	\item $U-L > 2$: If $|\mountain|\geq |\valley|$, choose an arbitrary subset $\mountain'\subset \mountain$ with $|\mountain'|= |\valley|$. Otherwise, choose an arbitrary subset $\valley'\subset \valley$ with $|\valley'|= |\mountain|$.
	\item $U-L = 2$: Note that mountains and valleys are on the same level, and $|\valley| \geq |\mountain|$ hold. Choose $\valley' \subset \valley$ to be the subset of size $|\mountain|$ with \emph{smallest} $x$-coordinates and set $\mountain' = \mountain$.
\end{itemize}

\begin{claim}
\label{claim:monotone}
	Phase~4.2 yields a stable schedule for reconfiguring all monotone configurations with the same number of robots into the same triangular configuration.
\end{claim}

\begin{claimproof}
	We show the following properties.
	\smallskip
	\begin{enumerate}
		\item The pushing paths do not cross. 
		\item After each iteration, $U-L$ decreases by at least 1. 
		\item The approach terminates in a configuration that only depends on the number of robots.
	\end{enumerate}
	
	For the non-crossing property of the paths, we show that the mountains and
	valleys have pairwise disjoint $x$- and $y$-coordinates. If $U-L > 2$, assume for the sake of
	contradiction that a mountain and a valley have equal $x$-coordinate. These
	positions have to be adjacent, which would result in $U-L = 1$, a
	contradiction. If $U-L = 2$ holds,
	mountains and valleys are already on the same level, so they have distinct
	$x$-coordinates. Similar arguments hold for the $y$-coordinates. It
	directly follows from the definition of levels  that for all positions $(x_1,
	y_1), (x_2, y_2) \in \mountain\cup \valley$, $x_1 < x_2$ implies $y_1 > y_2$. Together with
	the construction of the paths by matching adjacent mountains and valleys, it
	follows that these paths do not cross.
	
	By construction, $U > L$ always holds. If 
	$\mountain' = \mountain$ holds, level
	$U-1$ gets empty; therefore,~$U$ is decreased. In any case, $U$ cannot be
	increased, because all positions that get occupied are at level $L+1$. A similar
	argument holds if 
	$\valley' = \valley$ holds, thus, $U-L$ decreases by
	at least 1 in each iteration.
	
	If $U-L = 2$ holds, there is at most one level that contains occupied
	and empty positions. All robots on this level will be positioned leftmost; all
	levels with greater values are empty, and all levels with lower values are
	full.  
	
	Because the paths are non-crossing, the computed schedule is stable.
\end{claimproof}

\begin{proof}[Proof of \cref{lem:single-tile-reconfig}]
	Compute stable schedules consisting of Phases~4.1 and~4.2 for the start and the target configuration, $C_s$ and $C_t$, to reconfigure them into~$C_{\Delta}$. Reverting $C_t \rightrightarrows_{\chi} C_{\Delta}$ and combining it with $C_s \rightrightarrows_{\chi} C_{\Delta}$ results in the entire stable schedule transforming $C_s$ into~$C_t$. Because Phase~4.1 as well as Phase~4.2 can be realized in a makespan of $\mathcal{O}(d)$ (because a robot's maximum distance to a side, and the maximum level difference are upper-bounded by~$\mathcal{O}(d)$), the makespan of  
	$C_s \rightrightarrows_{\chi} C_t$ is in $\mathcal{O}(d)$.
\end{proof}

\subsubsection{Refilling Tiled Configurations}\label{sec:realrefilling}

Now we describe the final step for reconfiguring a tiled start configuration $C_s'$ into a tiled target configuration $C_t'$; because no robots are destroyed or created, this hinges on shifting robots between adjacent tiles, such that the required filling levels are achieved.

\begin{lemma}\label{lem:correctness_refilling}
	We can efficiently compute a stable schedule transforming $C_s'$ into~$C_t'$ within a makespan of $\mathcal{O}(d)$.
\end{lemma}

\subparagraph*{Outline of the Refilling Phase.} To compute the schedule of~\cref{lem:correctness_refilling}, we transfer robots between tiles, so that each tile $T$ contains the desired number in~$C_t'$. 
We model this robot transfer by a \emph{supply and demand flow}, see~\cref{fig:refilling_overview,fig:flowb}, followed by partitioning the flow into~$\mathcal{O}(1)$ subflows, such that each subflow can be realized within a makespan of $\mathcal{O}(d)$. 
For realizing a single subflow, we use the approach of~\cref{sec:reconfig-single-tiles} as a preprocessing step, i.e., to rearrange robots participating in a specific subflow and place them at suitable positions.

\begin{figure}[htb]
	\centering
	\includegraphics[page = 1, scale=0.60]{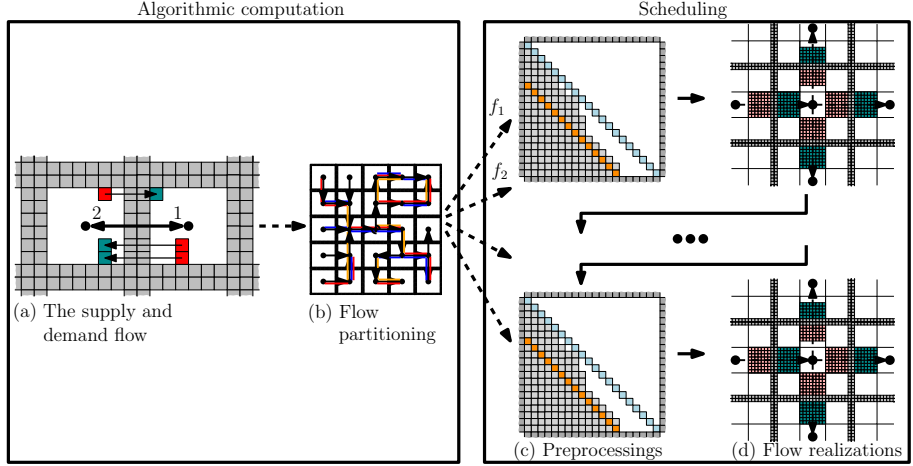}
	\caption{An overview of the schedule refilling tiles: transforming $C_s'$ into $C_t'$ by realizing a partition of a supply and demand flow that is computed in advance.}
	\label{fig:refilling_overview}
\end{figure}

\subparagraph*{Modeling Transfer of Robots via a Supply and Demand Flow.} We model the transfer of robots between tiles as a flow $F: E(G) \rightarrow \mathbb{N}$, using the directed graph $G = ({\mathcal{T}}, E)$ which is dual to the tiling ${\mathcal{T}}$.  
Let $B$ be the bottleneck matching between vertices from the original (non-tiled) $C_s$ and vertices from the final (non-tiled) $C_t$.
In $G$ we have an edge $(u,v) \in E$, if there is at least one matching edge $(r_u,r_v) \in B$, such that $r_u$ lies in the interior of the tile~$u$ in configuration $C_s'$, and $r_v$ lies in the interior of the tile~$v$ in configuration~$C_t'$. 
The flow value $F((u,v))$ of $(u,v)$ is equal to the number of such edges $(r_u,r_v) \in B$. 
A vertex $v \in V(G)$ has a \emph{demand} of $a > 0$ if the sum of the flow values of outgoing edges from $v$ plus $a$ is equal to the sum of the flows of incoming edges to $v$.
Analogously, $v$ has a \emph{supply} of $a > 0$ if the sum of flow values incoming to~$v$ plus $a$ equals the sum of flow values outgoing from $v$.

\begin{figure}[htb]
	\centering
	\includegraphics{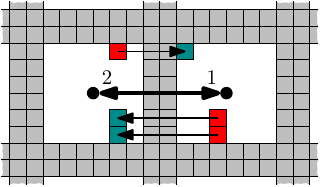}
	\caption{A pair of configurations $C_s'$ (red) and $C_t'$ (cyan) and the resulting flow modeling the edges of the bottleneck matching. The left vertex of $G$ has demand $1$, the right vertex has supply $1$.}
	\label{fig:flowb}
\end{figure}

\subparagraph*{Flow Partition and Algorithmic Computation.} We now define a \emph{flow partition} of $F$.
For $k \in \mathbb{N}$, a \emph{$k$-subflow of $F$} is a supply and demand flow $f: E(G) \rightarrow \mathbb{N}$ on $G$ with $f(e) \leq \min \{ k, F(e) \}$. 
A \emph{$k$-partition} of the flow $F$ is a set $\{ f_1,\dots,f_{\ell} \}$ of $k$-subflows of $F$, such that $\sum_{i = 1,\dots,\ell} f_i(e) = F(e)$ holds for all edges $e \in E(G)$.

We describe our approach for computing a $\phi$-partition of $F$, with $\phi := \lfloor \frac{(cd-2)^2}{9} \rfloor$; the value $\phi$ arises from partitioning the interior (made up of $(cd-2)^2$ pixels) of each tile into~$9$ almost equally sized \emph{subtiles} that are used for realizing a single set of paths as described below, see~\cref{fig:tunnel_aproach}. 

\begin{figure}[ht]
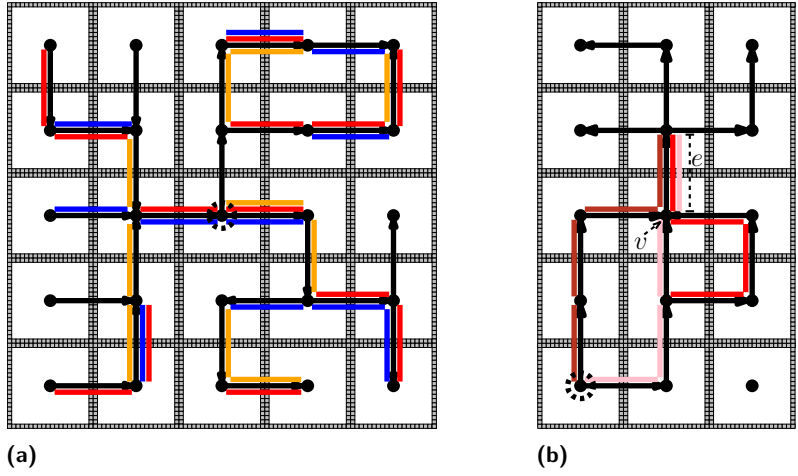

	\centering
	\begin{subfigure}{0.3\textwidth}
		\includegraphics[scale=0.35, page=1]{graphics/flow_partition2}
		\caption{}
		\label{fig:flow_partition-a}	
	\end{subfigure}\hfil
	\begin{subfigure}{0.3\textwidth}
		\includegraphics[scale=0.35, page=2]{graphics/flow_partition2}
		\caption{}
		\label{fig:flow_partition-b}	
	\end{subfigure}\hfil
	\caption{(a) We model the movements of robots between tiles as paths
		forming a tree. By greedily assigning these paths to sets (here highlighted by different colors), such that inside each set each edge is contained in no more than
		$3$ paths, we obtain that at most $\Theta(d^2)$ sets are needed. (b)~An example of a set containing three paths (dark red, pink, and red; assigned in that order to~$S_j$) having a common edge $e$ caused by a vertex $v$ of $e$ with an incoming degree of~$3$.}
	\label{fig:flow_partition}	
\end{figure}

We compute a $1$-partition of $G$, with each
$1$-subflow being either a cycle or a path that connects a supply vertex with a
demand vertex. Because the robots are unlabeled, we can simplify $G$ by
eliminating all cyclic $1$-subflows, as they are not necessary to realize
this specific transfer of robots; 
note that this also applies to bidirectional edges. Furthermore, we
replace diagonal edges $(v,w) \in E(G)$ by a pair of adjacent edges
$(v,u),(u,w) \in E(G)$. After all cyclic
subflows are removed, $G$ is a planar, directed forest consisting of $1$-subflows
that are paths, see \cref{fig:flow_partition-a}.  
We process each tree $A\subset G$ that is made up of paths $P_1, P_2, \dots$ separately as follows: We choose an arbitrary vertex of $A$ as its
\emph{root} and consider the \emph{link distance} of $P_i$ as the minimal
length of the path between a vertex from $P_i$ and the root of $A$. Let
$P_1,P_2, \dots \subseteq G$ be sorted by increasing link distances, which is important for our next argument, regarding that we can partition these paths into constant many subflows each of which is realizable in time linear in $d$. We greedily assign each path $P_i = P_1,P_2,\dots$ to a set~$S_j$, such that the first edge $e_1$ of $P_i$ is not part of another path inside~$S_j$. If no such a set~$S_j$ exists, we create a new set $S \leftarrow \{
P_i \}$. For each tree we use the same sets $S_1,S_2,\dots$ of
collected paths, because different trees of $G$ are disjoint. Note, that the construction of the sets $S_j$ allow that an edge is part of at most three paths inside $S_j$. This is due to the fact that the income degree of the head vertex of a directed edge is at most three in the setting of a grid graph, resulting in at most three outgoing edges.

Finally, we greedily partition $\{ S_1,S_2, \dots \}$ into subsets ${G}_1, {G}_2, \dots$ called \emph{groups}, made up of $\frac{(cd-2)^2}{9 \cdot 3}$ sets. For each group $G_i$, we define a subflow $f_i$ by setting
$f_i(e)$ as the number of paths from $G_i$ containing the edge $e$. As for each set $S_i$ and each edge $e$, there are at most three paths inside $S_i$ containing $e$, each resulting subflow $f_i$ is a $\left( \frac{(cd-2)^2}{9 \cdot 3} \cdot 3 \right)=\phi$-subflow. Finally, we have to upper-bound
the number of resulting subflows, i.e., the number of groups.

\begin{lemma}\label{lem:flow_partition}
	The constructed $\phi$-partition $\{f_1,f_2,\dots\}$ consists of at most $28$ subflows.
\end{lemma}
\begin{proof}
	The number of robots in the interior of each tile is upper-bounded by
	$(cd-2)^2$ which is the number of pixels in a tiles' interior. Substituting at most two diagonal edges increases
	the number of robots leaving or entering a specific tile by at most
	$2(cd-2)^2$. Hence, the flow value of each edge is upper-bounded by $3(cd-2)^2$ $(\star)$
	after substituting diagonal edges.
	
	Suppose our approach constructs $28$  groups implying that our approach gets into configuration with $27 \cdot \frac{(cd-2)^2}{9 \cdot 3} =: \lambda$ constructed sets $S_1,S_2, \dots $ and the current path $P_i$ has to be
	assigned to a new set $S_{\lambda +1}$. This implies that each set $S_1,\dots,S_{\lambda}$ contains a path $S_i$ containing the first edge $e_1$ of $P_i$. Hence,
	there are at least $27 \cdot \frac{(cd-2)^2}{9} + 1 > 3(cd-2)^2$
	paths containing $e_1$, contradicting $(\star)$.  
\end{proof}

\subparagraph*{Realizing a $\phi$-Partition.} Now we describe how to reconfigure a tiled configuration, such that a $\phi$-subflow is removed from $G$. 
A $\phi$-subflow $f_i$ is \emph{realized} by transforming the current configuration into another configuration, such that for each edge $e$ with $f_i((T,T')) > 0$, the number of robots in the interior of tile $T$ is decreased by $f_i((T,T'))$ and the number of the robots in the interior of tile~$T'$ is increased by $f_i((T,T'))$. 


\begin{figure}[ht]
	\centering
	\begin{subfigure}[b]{0.42\linewidth}
		\includegraphics[scale=0.68, page=1]{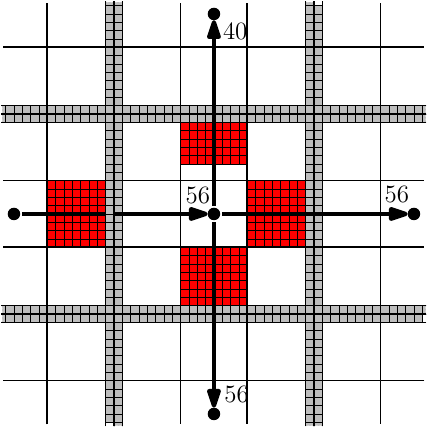}
		\caption{}
		\label{fig:tunnel_aproach-a}
	\end{subfigure}\hfil
	\begin{subfigure}[b]{0.42\linewidth}
		\includegraphics[scale=0.68, page=2]{graphics/tunnel_approach}
		\caption{}
		\label{fig:tunnel_aproach-b}
	\end{subfigure}\hfil
	\begin{subfigure}[b]{0.12\linewidth}
		\caption{\vspace*{0pt}}
		\label{fig:tunnel_aproach-c}
	\end{subfigure}\hfil
	\caption{(a)~A portion of a set of paths containing a vertex $v$: $56$
		paths are passing $v$ from left to right, while $40 + 56 = 96$ paths are starting
		in $v$. (b)~The configuration after realizing the set of paths shown in
		the previous figure. (c)~How positions of robots have changed after
		realization.} 
	\label{fig:tunnel_aproach}
\end{figure}

Next we realize a specific $\phi$-subflow within a makespan of
$\mathcal{O}(d)$. In particular, we partition the interior of each tile $T$
into $9$ \emph{subtiles} with equal side lengths (up to rounding), see
\cref{fig:tunnel_aproach}. For each subflow $f_i$, we place
$f_i((T,T'))$ robots inside the middle subtile of $T$ that shares an edge with 
the boundary of $T$ adjacent to $T'$. In particular, robots placed in the same
subtile are arranged in layers of width $\lfloor \frac{cd-2}{9} \rfloor$ as
close as possible to the boundary of the tile, see
\cref{fig:tunnel_aproach-a}. The resulting arrangement of robots
inside the subtile of $T$ is a \emph{cluster} and~$T'$ the \emph{target tile} of the
cluster. By a single application of the approach from
\cref{sec:reconfig-single-tiles}, all clusters of all tiles are arranged
simultaneously within a makespan of $\mathcal{O}(d)$. Finally, simultaneously
pushing all clusters of all tiles into the direction of their target tiles realizes~$S_i$, see \cref{fig:tunnel_aproach-b}. Note that not
all robots are pushed into the target tile $T'$ but some replace robots
on the boundaries between $T$ and $T'$, see \cref{fig:tunnel_aproach-c}. 

Repeating this approach for each subflow $f_i$ leads to a stable schedule
that realizes the entire flow $F$ within a makespan linear in $d$, i.e.,
transforms $C_s'$ into~$C_t'$ within $\mathcal{O}(d)$ transformation steps,
see \cref{fig:refilling_overview}. 

\begin{proof}[Proof of \cref{lem:correctness_refilling}]
	The schedule described above (and illustrated in
	\cref{fig:tunnel_aproach}) is correct because of the following. 
	For each tile $T \in V(G)$, the number of the paths starting in
	$T$ equals the supply of $T$, while the number of paths ending in $T$ equals
	the demand of $T$. As $T$ has either a supply or a demand, but not
	both, so the tile $T$ is never
	too full to receive additional robots, and never too empty to deliver robots as
	required by paths containing~$T$. 
\end{proof}

This concludes the proof of \cref{thm:scaled-instances-bounded-stretch}. Finally, we adapt the result to the general case, for which overlap of the start and target configuration is not guaranteed.

\begin{corollary}\label{cor:non-overlapping-instances}
	There is a constant $c^*$ such that for any pair of start and target configurations with a scale of at least $c^*$, there is a stable schedule of constant stretch.
\end{corollary}

\begin{proof}
	In case of a pair $(C_s, C_t)$, consisting of a start and a target configuration, that does not overlap, our algorithm computes in a first step a minimum bottleneck matching mapping $C_s$ to $C_t$ resulting in a bottleneck distance $\overline{d}$, and translates $C_s$ into a configuration $\overline{C}_s$ overlapping the target configuration within a makespan of $\overline{d}$. This results in a bottleneck distance $d$ between $\overline{C}_s$ and $C_t$ which is at most $2\overline{d}$. As \cref{thm:scaled-instances-bounded-stretch} guarantees a makespan linear in $d$, we obtain a makespan linear in $\overline{d}+d$, i.e., linear in $\overline{d}$ for the overall algorithm.
\end{proof}

As the diameter of the pair $(C_s,C_t)$ is a lower bound for the makespan of any schedule transforming $C_s$ into $C_t$, we obtain the following.

\begin{corollary}\label{cor:scaled-instances-constant-factor-approx}
	There is a constant-factor approximation for computing stable schedules with minimal makespan 
	between pairs of start and target configurations with a scale of at least $c^*$, for some constant $c^*$.
\end{corollary}

All steps in the overall algorithm can be computed in an efficient manner. Critical is the computation of matchings (Phase 2): the underlying geometry can be exploited for computing bottleneck matchings for a set of points in the plane in $\mathcal{O}(n^{1.5}\log n)$~\cite{efrat2001geometry}. The~required BFS tree (Phase 3) and flow decomposition~(Phase 4) can both be computed in straightforward manner; note that Phase 4 does not require a smallest decomposition, and the edge capacities are strongly polynomial in the number $n$ of robots.

%% file: 05-conclusion-full.tex
\section{Conclusion}
We have shown that connected coordinated motion planning is challenging even in relatively simple cases, such as unlabeled robots that have to travel a distance of at most 2 units, by establishing \NP-completeness.
On the other hand, we have shown that (assuming sufficient scale of the swarm), it is possible to compute efficient reconfiguration schedules with constant~stretch.

It is straightforward to extend our approach to other scenarios, e.g., to three-dimensional configurations. 
Other questions appear less clear. 
Is it possible to achieve constant stretch for arrangements with very small scale factor? 
We believe that this may hinge on the ability to perform synchronized shifts on long-distance ``chains'' of robots without delay, which is not a valid assumption for many real-world scenarios. 
(A well-known example is a line of cars when a traffic light turns green.)
As a consequence, the answer may depend on crucial assumptions on motion control; we avoid this issue in our approach. 
Can we provide alternative approaches with either weaker scale assumptions or better stretch factors? 
Can we extend our methods to the labeled case? 
All these questions are left for future work.